\documentclass[11pt]{article}
\usepackage{graphicx,euscript,verbatim}
\usepackage{amssymb,url,amsmath,amsthm}
\usepackage{natbib}

\theoremstyle{plain}

\newtheorem{Thm}{Theorem}[section]

\newtheorem{Lem}[Thm]{Lemma}

\newcommand{\nc}{\newcommand}

\nc{\regpi}{\pi}
\nc{\U}{\mathcal{U}}
\nc{\V}{\mathcal{V}}
\nc{\tU}{\check{U}}
\nc{\tV}{\check{V}}
\nc{\tN}{\bar{N}}
\nc{\tY}{\wt{Y}}
\nc{\wt}{\widetilde}
\newcommand{\tX}{\bar{X}}
\newcommand{\te}{\tilde{e}}
\newcommand{\Ex}{\mathbb{E}}

\nc{\on}{\operatorname}
\nc{\rec}{\frac{1}}
\nc{\cond}{\, \bigl|\,}
\nc{\N}{\mathbb{N}}
\newcommand{\M}{\mathcal{M}}
\normalsize

\nc{\tP}{\wt{P}}
\renewcommand{\P}{\mathbb{P}}
\nc{\tq}{\mathring{q}}
\nc{\tZ}{\wt{Z}}
\nc{\iid}{i.i.d.\ }
\nc{\R}{\mathbb{R}}
\nc{\eu}{\EuScript}
\nc{\indic}{\mathbf{1}}

\begin{document}
\title{Derivatives of the Stochastic Growth Rate}
\author{David Steinsaltz\\Department of Statistics\\University of Oxford\\1 South Parks Road\\Oxford OX1 3TG\\United Kingdom\\steinsal@stats.ox.ac.uk
\and Shripad Tuljapurkar \\ Biology Department \\Stanford University \\Stanford ,CA 94305 \\ USA
\and Carol Horvitz \\Biology Department \\ University of Miami \\ P.O. Box 249118\\
Coral Gables, FL 33124\\USA}

\maketitle
\newpage
\abstract{We consider stochastic matrix models for population driven by random environments which form a Markov chain. The top Lyapunov exponent $a$, which describes the long-term growth rate, depends smoothly on the demographic parameters (represented as matrix entries) and on the parameters that define the stochastic matrix of the driving Markov chain. The derivatives of $a$ --- the ``stochastic elasticities'' --- with respect to changes in the demographic parameters were derived by \cite{tuljapurkar1990pdv}. These results are here extended to a formula for the derivatives with respect to changes in the Markov chain driving the environments. We supplement these formulas with rigorous bounds on computational estimation errors, and with rigorous derivations of both the new and the old formulas.}

\section{Introduction}
Stochastic matrix models for structured populations are widely used in evolutionary biology, demographic forecasting, ecology, and population viability analysis  (e.g., \cite{tuljapurkar1990pdv, lee1994spf, morris2002qcb, caswell2001mpm, lande2003spm}). In these models, a discrete-time stochastic process drives changes in environmental conditions that determine the population's stage-transition rates (survival, fertility, growth, regression and so on). Population dynamics are described by a product of randomly chosen population projection matrices. In most biological situations the population's stage structure converges to a time-varying but stable structure \cite{cohen1977erg2}, and in the long run the population grows at a stochastic growth rate $a$ that is not random and is the leading Lyapunov exponent of the random product of population projection matrices \cite{furstenberg1960pro, cohen1977erg2, lange1979css, lange1981ssp, tuljapurkar1980pdv}. This growth rate $a$ is of considerable biological interest, as a fitness measure for a stage-structured phenotype \cite{tuljapurkar1982pdv}, as a determinant of population viability and persistence \cite{tuljapurkar1980pdv, morris2002qcb, lande2003spm}, and in a variety of invasion problems in evolution and epidemiology \cite{metz1992sho}.

The map between environments and projection matrices describes how phenotypes change with environments, i.e., the phenotypic norm of response, and we are often interested in how populations respond to changes in, say, the mean or variance of the projection matrix elements. Such questions are answered by computing the derivatives of $a$ with respect to changes in the projection matrices, using a formula derived by \cite{tuljapurkar1990pdv}. \cite{tuljapurkar2003mgr} called these derivatives stochastic elasticities, to contrast with the elasticity of the dominant eigenvalue of a fixed projection matrix to the elements of that matrix \citep{caswell2001mpm}. Stochastic elasticity has been used to examine evolutionary questions \citep{haridas2005eve} and the effects of climate change \citep{morris2008lcb}. At the same time, $a$ is also a function of the stochastic process that drives environments. Many processes, such as climate change \citep{boyce2006div}, will result in changes in the frequencies of, or the probabilities of transition between, environmental states. William Morris (personal communication 2005) posed the question: how is $a$ affected by a change in the pattern and distribution of environments, rather than by a change in the population projection matrices? To answer his question, we consider a model in which the environment makes transitions among one of several discrete states, according to a Markov chain.  Then what we want is the derivative of $a$ with respect to changes in the transition probabilities of this Markov chain.  This derivative exists (at least away from the boundaries of the space of stochastic matrices), and in fact we know from \cite{peres1992adl} that $a$ is an analytic function of the parameters of both the projection matrices and the parameters defining the stochastic matrix, in an open neighborhood of the set of stochastic matrices. In deterministic models \citep{caswell2001mpm}, the growth rate is represented as $\lambda = e^r$; then sensitivities are derivatives of the form $(\partial \lambda/\partial x)$ with respect to a parameter $x$ whereas elasticities are proportional derivatives of the form $(\partial r/\partial \log x)$. In stochastic models we compute derivatives of $a$, and these can be used to compute elasticities (as in \cite{tuljapurkar2003mgr}) or sensitivities.

Our first contribution here is a new formula for computing the derivative of $a$ with respect to changes in the transition probabilities of the environmental Markov chain. To obtain this result we show how an initial environmental state affects future growth, using coupling and importance sampling; this analysis may be of independent interest. Even with a formula in hand we must compute derivatives of $a$ by numerical simulation which is subject to both sampling (Monte Carlo) error and bias. Our second contribution here is to show how one can bound these estimation errors. Our third contribution is a rigorous proof of the heuristically derived formula given by \cite{tuljapurkar1990pdv} for the derivatives of $a$ to the elements of the population projection matrices.

In Section 2 of this paper we set out the model and assumptions, the approach to finding derivatives, along with necessary facts about the convergence of population structures and distributions. In Section 3 we discuss systematic and sampling errors and show how we can bound them. We illustrate this approach in Section 4 by presenting bounds (in Theorem 1) for simulation estimates of the stochastic growth rate $a$ and (in Theorem 2) for the derivatives of $a$ with respect to projection matrix elements. In Section 5 we define a measure of the effect of an initial environmental state on subsequent population growth and show how to estimate this measure using coupling arguments. Section 6 presents (in Theorem 4) the formula, algorithm, and error bounds for the derivative of $a$ with respect to the elements of the Markov chain that drives environments. We end by discussing how these theorems can be applied and some related issues concerning parameter estimation in such models. Proofs are in the Appendix.

\section{Model, Convergence and Stationary Distributions} \label{sec:obs}

We consider a population whose individuals exist in $K$ different stages (these may be, for example, ages, developmental stages or size classes). Newborns are in stage $i=1$. The progression between stages occurs at discrete time intervals at rates that depend on the environment in each time interval. The environment $e_t$ in period $t$ is in one of $M$ possible states; we denote the set of possible environments by $\M=\{1,\dots,M\}$. Individuals in stage $i$ at time $t$ move to stage $j$ at a rate $X_{e_{t+1}}(j,i)$. These rates are elements of a nonnegative population projection matrix, and at time $t$ when the environment is $e_t$ this matrix is denoted by $X_{e_t}$; there are $M$ such matrices, one for each environmental state.  We assume that allocation of individuals to classes and the identification of environment states are certain. We also assume that the total number of individuals in the population is large enough that we can ignore sampling variation. Successive environments are chosen according to a Markov process with transition matrix $P$ whose elements are $P(e,e')$ and whose stationary distribution is $\nu = \{\nu(e))\}$. We follow the standard convention for Markov chains, that $P(e,e')$ represents the probability of a transition from state $e$ to state $e'$; note that this is the opposite of the convention used in matrix population models. In some places we specialize to the case when the environments are \iid (independent and identically distributed), with distribution $\nu$.

To guarantee demographic weak ergodicity (Cohen 1977) we assume that
\begin{enumerate}
\item Each row of each population projection matrix has at least one positive entry.
\item There exists some $R>0$ such that any product $X_{e_{1}}\cdots X_{e_{R}}$ has all entries positive.
\item In the case of \iid distributions, all $\nu_{e}$ are positive. In the Markov case, the chain is assumed to be ergodic (so transitive and aperiodic), and environments are in the stationary distribution  of $P$, which will also be denoted by $\nu(e)$.
    \end{enumerate}

The population in year $t$ is represented by a vector $N_{t}\in \R_{+}^{K}$ with $N_t =\{(n_t(1),\dots,n_t(k))^{T}:\, n_t(i)> 0\}$. The superscript $T$ will always mean transpose; here it indicates that population vectors are column vectors. (There may be population classes early on that have 0 members; condition (ii) above forces all classes eventually to have positive membership, and so we assume without loss of generality that we start with all population classes occupied.) The population structure changes according to $N_{t+1}=X_{e_{t+1}}N_{t}$, and
\begin{equation} \label{E:basiceq}
N_{t}=X_{e_{t}}X_{e_{t-1}}\cdots X_{e_{1}}N_{0}.
\end{equation}
The normalized population structure $\tN_{t}:=N_{t}/(\sum_{i} N_{t}(i))$ does not converge to a fixed limit (as it would if the environment were constant) but it does converge in distribution. The long-run growth rate is not random and for each stage $i$,
$$
a:=\lim_{t\to\infty} t^{-1}\log N_{t}(i)=\lim_{t\to\infty} t^{-1} \log\frac{\sum_{i} N_{t}(i)}{\sum_{i} N_{0}(i)}
$$
exists, and is the same in every realization. This $a$ is called the ``stochastic growth rate''.

Counting the $K^{2}$ parameters in each matrix, there are at most $\left(M K^{2}+M^{2}\right)$ parameters. While all parameters must be nonnegative, and all elements of the population projection matrices but the birth rates must be $\le 1$, the only universal constraint is $\sum_{e\in\M}P(e',e)=1$. (There may, however, be further constraints imposed, as some transitions may be impossible. If we are considering age-structured populations, the matrices are Leslie matrices, each with only $2K-1$ potentially nonzero parameters.) The sensitivities we examine are derivatives of $a$ with respect to these parameters. When evaluating sensitivities we work in an explicit basis in which perturbations are described by an appropriate matrix and we refer to change ``in the direction of'' that matrix. This will be made precise in the analyses that follow.

\subsection{Convergence}  \label{sec:convergence}
We denote the $K$-dimensional column vector with 1's in all places by $\indic$. By default we use the $L^{1}$ norm $\|x\|:= \sum|x_{i}|$ when $x$ is a vector in $\R_{+}^{K}$, and write $\|X\|:=\|X\indic\|=\sum_{i,j=1}^{K}X_{i,j}$ when $X$ is a $K\times K$ matrix. Our assumptions imply that there are positive constants $\hat{k}$ and $\hat{r}$, such that for any environments $e_{1},\dots,e_{m}$,
\begin{equation} \label{E:normbound2}
\bigl| \log \|X_{e_{m}}\cdots X_{e_{1}}\|\bigr| \le \hat{k}+m\hat{r}
\end{equation}

We use the Hilbert projective metric $\rho$, described in \cite{golubitsky1975cas} and defined by
\begin{equation} \label{E:defrho}
\rho(x,y):=\log \max_{1\le i\le K} x(i)/y(i)+ \log \max_{1\le i\le K} y(i)/x(i).
\end{equation}
This is a pseudometric on $\R_{+}^{K}$ that is a metric on $\eu{S}:=\{(x(1),\dots,x(k))^{T}:\, x(i)> 0$ and $\sum x(i)=1\}$. The distance between two vectors is defined by the ray from the origin; that is, $\rho(x,y)=\rho(x/\|x\|,y/\|y\|)$. It has been shown by \cite{bushell1973hsm} that
$$
\frac{1}{2}\Bigl[\min\{x(i)\}+\min\{y(i)\}\Bigr] e^{-\rho(x,y)}\,\rho(x,y)\le \|x-y\|\le e^{\rho(x,y)}-1
$$
for any $x,y\in \eu{S}:=\{(x(1),\dots,x(k))^{T}:\, x(i)> 0$ and $\sum x(i)=1\}$. (The bound is actually shown with respect to the Euclidean norm, but the same argument holds for any $L^{p}$ norm.)  Thus, convergence in the projective metric implies convergence of the projections onto  $\eu{S}$ in the standard norms.

Following Lemma 1 of \cite{lange1979css}, we define a compact convex subset $\U\subset\eu{S}$ which is stable under the transformations $u\to X_{e}u/\|X_{e}u\|$ for any $e\in\M$ and includes the vector $\indic/K$ as well as all vectors of the form $X_{e_{R}}\cdots X_{e_{1}} y/\|X_{e_{R}}\cdots X_{e_{1}} y\|$ where $y\in \eu{S}$; and a compact convex subset $\V\subset\eu{S}^T$ which is stable under the transformations $v^T \to v^T X_{e}/\|v^T X_{e}\|$ and includes the vector $\indic^{T}/K$ as well as all vectors of the form $y^T X_{e_{R}}\cdots X_{e_{1}}/\|y^T X_{e_{R}}\cdots X_{e_{1}}\|$, where $y^T\in \eu{S}^T$. A useful fact about this metric follows: if $u,u'\in\eu{S}$, and $v^T$ any positive row vector, then
$$
\log \min\left\{\frac{u(i)}{u'(i)}\right\} \le\log \frac{v^T u}{ v^T u'}\le \log \max\left\{\frac{u(i)}{u'(i)}\right\}.
$$
Since $u,u'\in \eu{S}$, it follows that the left-hand side is $\le 0$ and the right-hand side $\ge 0$, so
\begin{equation} \label{E:ratiobound}
\bigl|\log v^T u-\log v^T u'\bigr| \le  \rho(u,u').
\end{equation}

From \cite[Theorem 2]{lange1981ssp} we know that there exist constants $k_{1},k_{2},r$, with $0<r<1$, such that for any $u,u'\in\U$ and environments $e_{1},\dots,e_{m}$,
\begin{equation} \label{E:basicconbound}
\rho\left(X_{e_{m}}\cdots X_{e_{1}}u,X_{e_{m}}\cdots X_{e_{1}}u'\right)\le k_{1}r^{m}\rho(u,u')\le k_{2} r^{m}.
\end{equation}
Of course, the constants may be chosen so that the same relation holds for the transposed matrices, with $u^T,(u')^T\in \V$.

It follows (as in Lemma 2 of \cite{lange1981ssp}) that for any $u,u'\in\U$ and environments $e,e_{1},\dots,e_{m}$,
\begin{equation} \label{E:rhobound}
\begin{split}
\left|\log\frac{\|X_{e}X_{e_{m}}\cdots X_{e_{1}}u\|}{\|X_{e_{m}}\cdots X_{e_{1}}u\|}-\log\frac{\|X_{e}X_{e_{m}}\cdots X_{e_{1}}u'\|}{\|X_{e_{m}}\cdots X_{e_{1}}u'\|}\right| &\le k_{1} r^{m} \rho(u,u')\\
&\le k_{2} r^{m}.
\end{split}
\end{equation}
The same relation holds when the matrices $X$ are replaced by their transposes, with $u^T,(u')^T\in \V$.

Since $\|X\|=\|X\indic\|$, and $\indic/K$ is in $\U$, it immediately follows that if $e'_{1},\dots,e'_{i}$ are any other environments,
\begin{equation} \label{E:normbound}
\left|\log\frac{\|X_{e}X_{e_{m}}\cdots X_{e_{1}}\|}{\|X_{e_{m}}\cdots X_{e_{1}}\|}-\log\frac{\|X_{e}X_{e_{m}}\cdots X_{e_{i+1}}X_{e'_{i}}\cdots X_{e'_{1}}\|}{\|X_{e_{m}}\cdots X_{e_{i+1}}X_{e'_{i}}\cdots X_{e'_{1}}\|}\right| \le k_{2} r^{m-i}.
\end{equation}

We note that the results in \cite{lange1981ssp} depend only on the set of matrices $X$ being compact, not on it being finite. In section \ref{sec:results} and beyond we will be letting the matrices $X_{e}$ and/or the transition matrix $P$ depend smoothly on a parameter $x$, which will take values either in $[-x_{0},x_{0}]$ or $[0,x_{0}]$. We may then choose the sets $\U$ and $\V$ and constants $k_{1},k_{2},r$ such that the properties above --- in particular, the stability of $\U$ and $\V$ and the bounds \eqref{E:basicconbound} and \eqref{E:rhobound} --- hold simultaneously for all values of the parameter.

\subsection{Time reversal and the stationary distribution} \label{sec:timerev}
The transition matrix for the {\it time-reversal} of $P$ will be denoted $\tP$, and is given by
$$
\tP(e,e'):= \frac{\nu(e')}{\nu(e)} P(e',e).
$$
A standard result (for example, see Theorem 6.5.1 of \cite{grimmett2001par}) tells us that if $e_{1},e_{2},\dots, e_m$ form a stationary Markov chain with transition matrix $P$, for a fixed $m$, the reversed sequence $e_{m},e_{m-1},\dots,e_{1}$ is a Markov chain with transition matrix $\tP$.

As described in \cite{lange1981ssp}, if $e_{1},e_{2},\dots$ forms a stationary Markov chain with transition probabilities $P$, there is a unique distribution $\regpi$ on $\U\times \M$ which is stable under the transformation $(u,e_{i})\mapsto (X_{e_{i}}u/\|X_{e_{i}}u\|,e_{i+1})$. That is, if the normalized population structure $\tN_{t}$ paired with $e_{t+1}$ is chosen from the distribution $\regpi$, then the pair $(\tN_{t+1},e_{t+2})$ will also be in the distribution $\regpi$. Furthermore,
\begin{enumerate}
\item For any initial population distribution $\tN_{0}$, and any initial environment $e_{0}$, the random pair $(\tN_{t},e_{t+1})$ converges in distribution to $\regpi$.
\item For a population distribution $u_{0}\in\U$, choose any random sequence of environments $e_{0},e_{1},\dots$. We may define a sequence of random vectors $U_{t}:=X_{e_{1}}\cdots X_{e_{t}}u_{0}/\|X_{e_{1}}\cdots X_{e_{t}}u_{0}\|$. Then $U_{t}$ converges pointwise to a random vector $U_{\infty}:=\lim_{t\to\infty} U_{t}$. If we identify $u_{0}$ with $\tN_{0}$, let the sequence of environments be realized from the time-reversed chain $\tP$, and $e_{0}$ from the stationary distribution $\nu$, then $(U_{t},e_{0})$ has the same distribution as $(\tN_{t},e_{t+1})$. Hence, the distribution of $(U_{\infty},e_{0})$ is $\regpi$. Furthermore, if $u_{0}\in \U$, we obtain directly from \eqref{E:basicconbound}
\begin{equation} \label{E:limbound}
\rho(U_{\infty},U_{t})\le k_{2} r^{t}.
\end{equation}
\end{enumerate}

The same holds true, of course, if we reverse the matrix multiplication: Starting from any nonnegative $M$-dimensional row vector $v^T_{0}$, we define from a sequence of environments $e_{0},e_{1},\dots$ the sequence of row vectors $V^T_t:=v^T_{0}X_{e_{t}}\cdots X_{e_{1}}/\|v^T_{0}X_{e_{t}}\cdots X_{e_{1}}\|\in \V$. Then $V_{t}$ converges pointwise to a random vector $V_{\infty}:=\lim_{t\to\infty} V_{t}$. When the sequence of environments has been chosen from the chain $P$, we denote the distribution of $(V_{\infty},e_{0})$ by $\tilde{\pi}$. As before, $\tilde{\pi}$ is the stationary distribution for the Markov chain on $\V\times\M$, defined by taking $(v^T,e_{t})$ at time $t$ to $(v^T X_{e_{t}}/{\|v^T X_{e_t}\|},e_{t+1})$ at time $t+1$, where the environments $(e_{0},e_{1},\dots)$ are drawn from the backward chain.

We  also define the regular conditional distributions $\regpi_{e}$ on $\U$ as follows: pick $(U,\mathbf{e})$ from the distribution $\regpi$, conditioned on $\mathbf{e}$ being $e$, and take $\regpi_{e}$ to be the distribution of $U$. Similarly we define $\tilde{\pi}_{e}$.
In the case of \iid environments, of course, $\regpi$ and $\tilde{\pi}$ are simply products of an independent population vector and environment; the environment has distribution $\nu$, and the stationary population distributions we  also denote (by an abuse of notation) by $\regpi$ and $\tilde{\pi}$.

\subsection{Estimating contraction rates}  \label{sec:concon}
The constants $r$ and $k_{1}$, defined in \eqref{E:basicconbound}, are crucial to the analysis at several stages. We describe here how to obtain plug-in estimates for these quantities. This is by no means the most efficient algorithm, nor does it obtain the best bounds, but it should be feasible for problems of moderate size. To begin, we let $Y_{1},\dots,Y_{S}$ be the collection of all products of the form $X_{e_{1}}X_{e_{2}}\cdots X_{e_{R}}$; here $S=M^{R}$. Following \cite{bushell1973hsm} we define
\begin{equation} \label{E:delta}
\Delta(Y):=\sup \bigl\{ \rho(Yx,Yx') : x,x'\in \on{Int}(\R_{+}^{K})\bigr\}.
\end{equation}
Since for any positive vectors $u,u',u''$ we have $\rho(u+u',u'')\le \max\{\rho(u,u''),\rho(u',u'')\}$, the maximum of $\rho(u,u')$ among vectors in a cone is taken between extreme points of the cone, we can compute $\Delta(Y)=\sup \bigl\{ \rho(Y^{(i)},Y^{(j)})\bigr\}$ where $Y^{(i)}$ and $Y^{(j)}$ are two columns of the matrix $Y$. By Theorem 3.2 of \cite{bushell1973hsm} it follows that if we let $r_{0}:=\max_{1\le i\le S} \tanh \Delta(Y_{i})/4$ (which is $<1$), then for all $x,x'\in \on{Int}(\R_{+}^{K})$, and any $e_{1},\dots,e_{R}\in\M$,
$$
\rho(X_{e_{1}}\cdots X_{e_{R}}u,X_{e_{1}}\cdots X_{e_{R}}u')\le r_{0}\rho(u,u').
$$
Thus, \eqref{E:basicconbound} holds with $r=r_{0}^{1/R}$ and $k_{1}=r^{1-R}$. Since we want the bounds to hold for the reversed products as well, we repeat these computations with $(X_{e})$ replaced by $(X_{e}^{T})$, and finally adopt the larger values of $r$ and $k_{1}$.

\section{Errors and How to Bound Them} \label{sec:main}
A standard approach to estimating $a$, and the derivatives that we give later on, is to choose a fixed starting vector $u_{0}\in \U$, simulate sequences $e_{0}(i),e_{1}(i),\dots,e_{m}(i)$ independently from the stationary Markov chain with transition probabilities $P$ ($i=1,\dots,J$), and then compute
\begin{equation} \label{E:approx}
\begin{split}
a_{m}&:= \Ex\left[ \log \frac{\|X_{e_{m}(i)}X_{e_{m-1}(i)}\cdots X_{e_{0}(i)} u_{0}\|}{\|X_{e_{m-1}(i)}\cdots X_{e_{0}(i)} u_{0}\|}\right]\\
&\approx\rec{J} \sum_{i=1}^{J} \Bigl(\log \|X_{e_{m}(i)}X_{e_{m-1}(i)}\cdots X_{e_{0}(i)} u_{0}\| \\
&\hspace*{4cm}- \log\|X_{e_{m-1}(i)}\cdots X_{e_{0}(i)} u_{0}\| \Bigr).
\end{split}
\end{equation}
It is important not only to know what would be an appropriate approximation to $a$ or its derivatives in the sense of being asymptotically correct, but also to have rigorous bounds for the error arising from any finite simulation procedure, such as \eqref{E:approx}. There are two sources of error: systematic error, arising from the fact that $X_{e_0}, u_{0}$ is not exactly a sample from the distribution $\regpi$; and sampling error, arising from the fact that we have estimated the expectation by averaging over a random sample.

\subsection{Systematic error} \label{sec:system}
By ``systematic error'' we mean the error in our estimate of $a$ arising from the difference between the distribution we are aiming for and the distribution we are actually sampling from. The quantity we are trying to estimate may be represented as $a=\pi[F]$, expectation of a certain function $F$ with respect to the distribution $\pi$. If we can simulate $Z_{1},\dots,Z_{J}$ from $\pi$, then $\hat a_{J}:=J^{-1}\sum_{j=1}^{J} F(Z_{j})$ is an unbiased estimator of $\pi[F]$, and will be consistent under modest assumptions on $F$ and the independence of the samples. Suppose, though that what we have are not samples from $\tilde{\pi}$, but samples $Z'_{j}$ from a ``similar'' distribution $\pi'$. Then we can bound the error by
\begin{equation} \label{E:bup0}
|a-\hat a|\le \left| J^{-1}\sum_{j=1}^{J} F(Z'_{j})-\pi'[F] \right| + \Bigl|\pi[F]-\pi'[F] \Bigr|.
\end{equation}
Here the first term on the right-hand side is the sampling error, and the second term is the bias, the expected value of systematic error. The problem is that the bounds we can obtain for the bias are likely to be crude, absent good computational tools for the distribution $\pi$ (and if we could compute analytically from $\pi$, we wouldn't need to be simulating). 

Alternatively, if we can couple the samples $Z'_{j}$ from the approximate distribution $\pi'$ to exact samples $X_{j}$ from the distribution $\pi$, we can break up the error in a slightly different way:
\begin{equation} \label{E:bup}
\begin{split}
|a&-\hat a|\\
&\le \Bigl| J^{-1}\sum_{j=1}^{J} F(Z'_{j})-\tilde{\pi}[F] \Bigr| + \Bigl|J^{-1}\sum_{j=1}^{J} F(Z_{j}) - J^{-1}\sum_{j=1}^{J} F(Z'_{j}) \Bigr|\\
&\le \Bigl| J^{-1}\sum_{j=1}^{J} F(Z'_{j})-\tilde{\pi}[F] \Bigr| + J^{-1}\sum_{j=1}^{J} \Bigl|F(Z_{j}) - F(Z'_{j}) \Bigr|\\
&\le \Bigl| J^{-1}\sum_{j=1}^{J} F(Z'_{j})-\tilde{\pi}[F] \Bigr| + J^{-1}\sum_{j=1}^{J} \underset{x}{\on{ess}\sup}\Bigr\{\Bigl|F(Z_{j}) - F(Z'_{j}) \Bigr|\Bigr\},
\end{split}
\end{equation}
where the essential supremum in the last line is taken over the distribution of $X'_{j}$ conditioned on $X_{j}$. Bounds for the sampling error in \eqref{E:bup0} will generally also be bounds for the first term in \eqref{E:bup}. The second term in \eqref{E:bup}, on the other hand, which takes the place of the bias, is a random variable, computed from the samples $X_{j}$. Its expectation is still a bound on the bias. The crucial fact is that the last line may be computable without knowing in detail what the ``correct'' sample $Z_{j}$ is.

A small disadvantage of this approach is that the systematic error varies with the sample. To achieve a particular fixed error bound we need an adaptive approach, whereby we successively extend our sequence of matrices until the error crosses the desired threshold. In keeping with our comment in section \ref{sec:marcoup}, we note here that this approach to estimating the systematic error in simulations is essentially just a version of the Propp-Wilson algorithm.

\subsection{Sampling error} \label{sec:sampling}
The sampling error is difficult to control with current techniques, because the distribution of the samples is so poorly understood --- the very reason why we resort to the Monte Carlo approximation in the first place. The best we can do for a rigorous bound is to use Hoeffding's inequality (see \cite{hoeffding1963pis}), taking advantage of crude bounds on the terms in the expectation. Hoeffding's inequality tells us that if $X_{1},\dots,X_{J}$ are \iid random variables such that $\alpha\le X_{i}\le \beta$ almost surely, then for any $z>0$,
\begin{equation} \label{E:hoeffding}
\P\left\{ \left|\rec{J}\sum X_{i} - \Ex[X] \right|>z\right\}\le 2e^{-2Jz^{2}/(\beta-\alpha)^{2}}.
\end{equation}
This is essentially the same bound that we would estimate from the normal approximation if the standard deviation of $X$ were $(\beta-\alpha)/2$. Of course, the standard deviation will be smaller than this, but we do not know how much smaller. An alternative approach then would be to use the bound $2\tau(z\sqrt{J}/\hat\sigma)$, where $\hat\sigma$ is the standard deviation of the simulated samples, and $\tau$ is the cumulative distribution function of the Student t distribution with $J-1$ degrees of freedom. This will be a smaller bound, in that sense ``better'', but not precisely true for finite samples, to the extent that the sample distribution is not normal. Generally we will want to fix $p_{0}$, the confidence level, and compute the corresponding $z$, which will be
\begin{equation} \label{E:hoeffding2}
z_{0}= (\beta-\alpha)\sqrt{-\rec{2J}\log (p_{0}/2)}.
\end{equation}
The corresponding expression for the estimate based on the t-distribution is
\begin{equation} \label{E:T2}
z_{0}=\frac{\hat\sigma}{\sqrt{J}} t_{1-p_{0}/2}(J-1),
\end{equation}
where $t_{p}(J-1)$ is the $p$ quantile of the Student T distribution with $J-1$ degrees of freedom; that is, if $T$ has this distribution then $P\{T>t_{p}(J-1)\}=p$.

\section{Growth Rate and Sensitivity to Projection Matrices} \label{sec:results}
We present here extensions of two known results. In these cases (and in later results) we start by defining an estimator that converges to the the quantity we desire, and follow that by bounds on the systematic and sampling errors, as well as an error bound for estimates from a simulation estimator. We  state our results on error bounds in the form ``The quantity $Q$ may be approximated by the expectation of $A$, with systematic error bounded by $B$ and sampling error bounded by $C(J,p)$.'' This means that if $A_{1},\dots,A_{J}$ are independent realizations of $A$, then the probability that the true value of $Q$ is not in the interval $J^{-1}\sum A_{i}\pm [B+C(J,p)]$ is no bigger than $p$. When describing an adaptive bound on the systematic error, $B$ will depend upon the particular simulation result. Again, the sampling error may be bounded either by a universally valid Hoeffding bound, based on known upper bounds on the samples, or by the t distribution using the standard deviation estimated from the sample, which provides a generally much superior bound, but which can only be treated as an approximation.

\subsection{Computing $a$}
The stochastic growth rate $a$ is commonly estimated by numerical simulation but, as discussed with examples by Caswell (2001), there is no general way to bound the errors in the estimated values. The following result provides suitable bounds.

\begin{Thm}\label{T:computeLE}
Let $u_{0}$ be any fixed element of $\,\U$, and $Y_{m}:=X_{e_{m}}X_{e_{m-1}}\cdots X_{e_{1}}$, where $e_{0},e_{1},\dots$ form a Markov chain with transition rates $P$. The stochastic growth rate may be approximated by the simulated expectation of
\begin{equation} \label{E:approxLE}
\log \frac{\|X_{e_{m+1}}Y_{m} u_{0}\|}{\|Y_{m} u_{0}\|},
\end{equation}
with systematic error bounded by $k_{2}r^{m}$ and sampling error at level $p$ on $J$ samples bounded by
\begin{equation} \label{E:samperrLE}
\left(\log \frac{\sup_{u\in\U} \max_{e\in\M} \|X_{e}u\|}{\inf_{u\in\U} \min_{e\in\M} \|X_{e}u\|} \right)\left(\frac{-\log p}{2J}\right)^{1/2}.
\end{equation}

When the simulated expectation is
$$
\rec{J} \sum_{i=1}^{J}\log \frac{\|X_{e_{0}}(i)Y_{m}(i) u_{0}\|}{\|Y_{m}(i) u_{0}\|}
$$
we may also bound the systematic error by
\begin{equation} \label{E:syserrLE}
\rec{J} \sum_{i=1}^{J} \sup_{u,u'\in\U} \rho\bigl(Y_{m}(i)u,Y_{m}(i)u'\bigr)\le \rec{J}\sum_{i=1}^{J} \Delta\bigl( Y_{m}(i)\bigr),
\end{equation}
where $\Delta$ is defined as in \eqref{E:delta}.
\end{Thm}

\nc{\plx}[1]{\frac{\partial a}{\partial #1}}

\subsection{Derivatives with respect to Projection Matrices}
We need care in defining derivatives of $a$ with respect to elements of the population projection matrices. As discussed in Tuljapurkar and Horvitz (2003) we must define how the matrix entries change, e.g., do we change fertility rates in a particular environment, or in all possible environments? Although the main formula here is known, Tuljapurkar's (1990) derivation did not justify a crucial exchange of limits (between taking the perturbation to zero and time to infinity). We provide a rigorous proof (see Appendix) and of course the error bounds here are new.

We will suppose that the matrices $X_{e}$ depend smoothly on a parameter $x$, so that we may define $\tX_{e}:=\partial X_{e}/\partial x$, and we define the base matrices to be at $x=0$. In some cases, the parametrization will be defined only for $x\ge 0$, and in those cases we will understand the partial derivatives to be one-sided derivatives, and the limits $\lim_{x\to 0}$ will be the one-sided limits $\lim_{x\downarrow 0}$.

\begin{Thm}\label{T:computeSS1}
Let $U_{e}$ and $V_{e}$ be independent random variables with distributions $\regpi_{e}$ and $\tilde{\pi}_{e}$. Then
\begin{equation} \label{E:exactSS1}
a':=\plx{x}=\sum_{e\in\M}\nu_{e}\Ex \left[\frac{V^T_{e} \tX_{e} U_{e}}{V^T_{e} X_{e} U_{e}}\right]
\end{equation}
Each term may be approximated by averaging samples of the form
\begin{equation} \label{E:approxSS1}
\sum_{e\in\M}\nu_{e}\frac{V^{(m)T}\tX_{e} U^{(m)}}{V^{(m)T} X_{e} U^{(m)}},
\end{equation}
where $u_0, v^T_0$ are any fixed elements of $\,\U,\, \V$ respectively, $U^{(m)}=X_{\te_{1}}\cdots X_{\te_{m}}u_{0}$ and $V^{(m)T}=v^T_{0}X_{e_{m}}\cdots X_{e_{1}}$,  $e=\te_{0},\te_{1},\dots,\te_{m}$ form a sample from the Markov chain $\tP$, and $e=e_{0},e_{1},\dots,e_{m}$ form a sample from the Markov chain $P$. The systematic error may be bounded uniformly by
\begin{equation} \label{E:unifboundSS1}
2(\exp(4k_{2} r^{m})-1) a',
\end{equation}
while the sampling error at level $p$ on $J$ samples is bounded by
\begin{equation} \label{E:samperrSS1}
2\sum_{e\in\M}\nu_{e}\sup_{u\in\U,v^T \in\V} \frac{v^T \tX_{e} u}{v^T X_{e}u} \left(\frac{-\log (p/2)}{2J}\right)^{1/2}.
\end{equation}

Suppose the simulated expectation is
$$
\sum_{e\in\M}\frac{\nu_{e}}{J} \sum_{j=1}^{J} \frac{(V^{(m)}(j))^T \tX U^{(m)}(j)}{(V^{(m)}(j))^T X_{e} U^{(m)}(j)},
$$
where
\begin{align*}
U^{(m)}(j)&=X_{\te_{1}(j)}\cdots X_{\te_{m}(j)}u_{0}=: \tY_{m}(j) u_{0}\text{ and}\\
(V^{(m)}(j))^T&=v^T_{0}X_{e_{m}(j)}\cdots X_{e_{1}(j)}=:v^T_{0}Y_{m}(j).
\end{align*}
Let
\begin{align*}
\U(j)&:=\tY_{m}(j) \U=\bigl\{ \tY_{m}(j)u\, :\, u\in \U\bigr\},\\
\V(j)&:=\V Y_{m}(j) =\bigl\{ v^T Y_{m}(j)\, :\, v^T \in \V\bigr\}.
\end{align*}
Then we may also bound the systematic error by
\begin{equation} \label{E:syserrSS1}
\begin{split}
&\sum_{e\in\M}\frac{\nu_{e}}{J} \sum_{j=1}^{J}\sup_{\begin{smallmatrix} u\in\U(j)\\v^T\in\V(j)\end{smallmatrix}} \left| \frac{v^T \tX_{e} u}{v^T X_{e} u} - \frac{(V^{(m)}(j))^T\tX_{e} U^{(m)}(j)}{(V^{(m)}(j))^T X_{e} U^{(m)}(j)} \right|\\
&\le \sum_{e\in\M}\frac{\nu_{e}}{J}\sum_{j=1}^{J} \left( \exp\left\{ 2\sup_{u\in\U(j)} \rho(u,U^{(m)}(j)) +2\sup_{v^T\in\V(j)} \rho(v,V^{(m)}(j)) \right\}-1\right)\\
&\hspace*{45mm}\times\frac{(V^{(m)}(j))^T\tX_e U^{(m)}(j)}{(V^{(m)}(j))^T X_{e} U^{(m)}(j)} \\
&\le \sum_{e\in\M}\frac{\nu_{e}}{J}\sum_{j=1}^{J} \frac{(V^{(m)}(j))^T\tX_e U^{(m)}(j)}{(V^{(m)}(j))^T X_{e} U^{(m)}(j)} \left( \exp\left\{ 2\Delta\bigl(Y_{m}(j)\bigr) + 2\Delta\bigl(Y_{m}(j)^{T}\bigr) \right\} -1 \right).
\end{split}
\end{equation}
\end{Thm}

Note that the bound \eqref{E:unifboundSS1} is given as a proportion of the unknown $a'$. It can be turned into an explicit bound by using an upper bound on $a'$. For instance, it is easy to compute that
\begin{equation} \label{E:Lip1}
\operatorname{Lip}(a)\le \sup_{u\in\U} \max_{e\in\M} \exp\left\{ \rho\left(\left|\tX_{e}u\right|,X_{e}u\right)\right\}
\frac{k_{2}}{1-r}
\end{equation}

\section{Environments and Coupling} \label{sec:marcoup}
Suppose we make a small change in the transition matrix $P$, and want to compare population growth along environmental sequences generated by the original and the perturbed matrix. We expect that the perturbed environmental sequences will only occasionally deviate from the environment that we ``would have had'' in the original distribution of environments. Computing the derivative of $a$ is then a matter of measuring the cumulative deviations due to these changes. In this section we take an essential first step: fix the transition matrix $P$ and compare cumulative change in total population size when starting in environment $e$, as compared with starting in the stationary distribution $\nu$. Variants of this problem arise in the standard Markov-Chain Monte Carlo (MCMC) problem: estimate by simulation the expectation of a certain function from the stationary distribution of a Markov chain when that distribution is unknown. We need to measure the distance between distributions, and hence the difference between expectations. A standard method for doing this is {\em coupling}. For an outline of coupling techniques in MCMC, see \cite{kendall2005nps} and \cite{RR04}. We use coupling in two ways, corresponding to the two components of the Markov chain: the environment and the population vector.

\nc{\eze}{{}_{e'}\zeta_{e}}
Fix environments $e$ and $e'$ (possibly the same). We define sequences $e_{0},e_{1},\dots$; $e'_{0},e'_{1},\dots$; and $\te_{0},\te_{1},\dots$: all three are Markov chains with transition probabilities $P$, but with $e_{0}=e$, $e'_{0}=e'$ and $\te_{0}$ having distribution $\nu$ (so that $(\te_{i})$ is stationary). Then
\begin{equation} \label{E:zeta}
\begin{split}
\eze&:=\lim_{t\to\infty} \Bigl(\Ex\left[ \log\|X_{e_{t}}\cdots X_{e_{0}}\|\right]- \Ex\left[ \log\|X_{e'_{t}}\cdots X_{e'_{0}}\|\right]\Bigr)\\
\zeta_{e}&:=\lim_{t\to\infty} \Bigl(\Ex\left[ \log\|X_{e_{t}}\cdots X_{e_{0}}\|\right]- \Ex\left[ \log\|X_{\te_{t}}\cdots X_{\te_{0}}\|\right]\Bigr)\\
&=\sum_{e'=1}^{\M} \nu_{e'} \cdot \eze.
\end{split}
\end{equation}
Note that when the environments are i.i.d.\ --- so $P(e,e')=\nu_{e'}$ --- we have
$$
\eze=\log\|V^{T}X_{e}\| - \log\|{V'}^{T}X_{e'}\| .
$$

Computing $\zeta_{e}$ depends on coupling the version of the Markov chain starting at $e$, to another version starting in the distribution $\nu$. We define the {\em coupling time} $\tau$ to be the first time such that $e_{\tau}=\te_{\tau}$; after this time the chains follow identical trajectories. If we know the distribution of $\tau$ and of the sequences followed by the two chains from time 0 to $\tau$, we can average the diferences in \eqref{E:zeta} to find $\zeta$. The advantage of coupling is, first, that it reduces the variability of the estimates, and second, that we know from the simulation when the coupling time has been achieved, which gives bounds on the error. A suitable choice is Griffeath's maximal coupling \citep{dG75} which we will apply in Pitman's \citep{pitman1976cmc} path-decomposition representation. (The coupling is ``maximal'' in the sense of making the coupling time, and hence the variance of the estimate, as small as possible.) However we must be careful about sampling values of $\tau$ because they may be large if the Markov chain mixes very slowly. To deal with this, we use a resampling technique to overweight coupling times that generate a large contribution to $\zeta$.

Beginning with a fixed environment $e$, the procedure is as follows:
\begin{enumerate}
\item Define the sequence of vectors $\alpha_{t}:=\bigl(\alpha_t(e')=P^{t}(e,e')-\nu(e')\bigr)$. We also define $\alpha^{+}_{t}$ and $\alpha^{-}_{t}$ to be the vectors of pointwise positive and negative parts respectively. Let $C(t)$ be any bound on $\bigl|\log\|X_{e_{1}}\cdots X_{e_{t}}\|-\log\|X_{e'_{1}}\cdots X_{e'_{t}}\|\,\bigr|$, where the $e_{i}$ and $e'_{i}$ are any environments. From \eqref{E:normbound2} we know that $2\hat{k}+2t\hat{r}$ is a possible choice for $C(t)$.
\item For pairs $(t,e')$, where $t$ is a positive integer and $e'\in\M$, define a probability distribution
$$
q(t,e'):= \begin{cases}
\nu_{e}&\text{if } e=e', t=0;\\
0&\text{if } e\ne e', t=0;\\
[\alpha_{t-1}^{+}P](e') - \alpha_{t}^{+}(e')&\text{otherwise.}
\end{cases}
$$
This is the distribution of the pair $(\tau,e_{\tau})$ for the maximally coupled chain. Define
$$
A:=\sum_{t_{*}=1}^{\infty}\sum_{e_{*}=1}^{M} q(t_{*},e_{*})C(t_{*}),
$$
and a probability distribution on $\N\times \M$
$$
\tq(t,e'):=\frac{q(t,e')C(t)}{A}.
$$

\renewcommand{\te}{\check{e}}
\item Average $J$ independent realizations of the following random variable: Let $(\tau,e')$ be chosen from the distribution $\tq$ on $\N\times \M$. Let $(e_{0}(\tau,e'),\dots,e_{\tau}(\tau,e'))$ and $(\te_{0}(\tau,e'),\dots,\te_{\tau}(\tau,e'))$ be a realization of the coupled pair of Markov chains with transition probabilities $P$ and starting at $e_{0}(\tau,e')=e$ and $\te_{0}(\tau,e')$ with distribution $\nu$, conditioned on the coupling time being $\tau$ and $e_{\tau}=\te_{\tau}=e'$. These realizations are generated from independent inhomogeneous Markov chains running backward, with transition probabilities
\begin{align*}
\P\bigl\{ e_{i-1}=x\cond e_{i}=y \bigr\} &=\frac{\alpha^{+}_{i-1}(x) P(x,y)}{\sum_{x'=1}^{M} \alpha^{+}_{i-1}(x') P(x',y)},\\
\P\bigl\{ \te_{i-1}=x\cond \te_{i}=y \bigr\} &=\frac{\alpha^{-}_{i-1}(x) P(x,y)}{\sum_{x'=1}^{M} \alpha^{-}_{i-1}(x') P(x',y)}.
\end{align*}
The random variable is then
$$
 Z:= \frac{A}{C(\tau)} \log\frac{\|X_{e_{\tau}(\tau,e')}\cdots X_{e_{0}(\tau,e')}\|}{\|X_{\te_{\tau}(\tau,e')}\cdots X_{\te_{0}(\tau,e')}\|}.
$$
(Note that the realizations corresponding to $\tau=0$ are identically 0. The possibility of $\tau=0$ has been included only to simplify the notation. In practice, we are free to condition on $\tau>0$.)
\end{enumerate}

The change from $q$ to $\tq$ is an example of importance sampling (cf. Chapter V.1 in \cite{asmussen2007sto}). We oversample the values of the random variable with high $\tau$ to reduce the variability of the estimate. The importance sampling makes $Z(j)$ a bounded random variable, with bound $A$. Imagine that we had a source of perfect samples $V^T(j)$ from the distribution $\tilde{\pi}_{e_{m}}^{*}$, and define
$$
\tZ(j):=\frac{A}{C(\tau(j))}\log\frac{\|V^T(j)X_{e_{\tau}(\tau(j),e'(j);j)}\cdots X_{e_{0}(\tau(j),e'(j);j)}\|}{\|V^T(j)X_{\te_{\tau}(\tau(j),e'(j);j)}\cdots X_{\te_{0}(\tau(j),e'(j);j)}\|}.
$$
Let
\begin{align*}
Y_{1}(j)&:= X_{e_{\tau}(j)}\cdots X_{e_{0}(j)},\\
Y_{2}(j)&:= X_{e_{\tau}(j)}\cdots X_{e_{t}(j)} X_{\te_{t-1}(j)}\cdots X_{\te_{0}(j)} .
\end{align*}
Then
\begin{equation} \label{E:syserrzetaet}
\bigr|\tZ(j)-Z(j)\bigr|\le \left( \Delta\bigl(Y_{1}(j)\bigr) + \Delta\bigl(Y_{2}(j)\bigr)\right).
\end{equation}
At the same time $\Ex[\tZ]=\zeta_{e}$, so we may use \eqref{E:hoeffding2} to compute the bound
\begin{equation} \label{E:zetasamplebound}
\P\left\{ \Bigl|\zeta_{e}-n^{-1}\sum_{j=1}^{n} \tZ(j)\Bigr|> 2A\sqrt{-\rec{2J}\log (p_{0}/2)}\right\} \le p_{0}.
\end{equation}

\renewcommand{\te}{\tilde{e}}
\begin{Lem} \label{L:zeta}
The limits defining the coefficients $\eze$ and $\zeta_{e}$ exist and are finite. We may approximate $\zeta_{e}$ by
\begin{equation} \label{E:approxzeta}
\rec{J}\sum_{j=1}^{J} \frac{A}{C(\tau(j))} \log\frac{\|X_{e_{\tau}(\tau(j),e'(j);j)}\cdots X_{e_{0}(\tau(j),e'(j);j)}\|}{\|X_{\te_{\tau}(\tau(j),e'(j);j)}\cdots X_{\te_{0}(\tau(j),e'(j);j)}\|}.
\end{equation}
If $p_{0}$ is any positive number, the probability is no more than $p_{0}$ that the error in this estimation is larger than
\begin{equation} \label{E:totalzetaerror}
\frac{1}{J} \sum_{j=1}^{J}  \left( \Delta\bigl(Y_{1}(j)\bigr) + \Delta\bigl(Y_{2}(j)\bigr)\right)+2A\sqrt{-\rec{2J}\log (p_{0}/2)},
\end{equation}
\end{Lem}

It remains to compute $A$. From standard Markov chain theory there is an $M \times M$ matrix $Q$ such that
$$
P=\indic \nu^T + Q,
$$
a constant $D >0$ and a number $\xi\in (0,1)$ such that
$$
\|Q^t\| \le D \xi^t,
$$
from which it follows easily that for any vector $v$,
\begin{equation} \label{E:spectralgap}
\|v^{T}P^{t}-\nu^{T}\|\le D \xi^t \|v\|.
\end{equation}

Then
$$
\alpha^T_{t}=\indic^T_{e} Q^t,
$$
where $\indic_{e}$ is the vector with 1 in place $e$ and $0$ elsewhere, and
$$
\|\alpha_{t}\|\le D \xi^t.
$$
If we use the bound $C(t)=\hat{k}+t\hat{r}$, then
\begin{equation} \label{E:Abound}
\begin{split}
A&=\sum_{t=1}^{\infty} \bigl(\hat{k}+t\hat{r}\bigr) \bigl( \|\alpha_{t-1}\| -\|\alpha_{t}\|\bigr)\\
&=\hat{k}\|\alpha_{1}\|+\hat{r}\sum_{t=1}^{\infty}\|\alpha_{t}\|\\
&\le D \xi \left( \hat{k}+\frac{D \hat{r}}{1-\xi} \right)
\end{split}
\end{equation}

\section{Derivatives with respect to Environmental Transitions}
\nc{\neps}{\nu^{(\epsilon)}}
\nc{\hneps}{\hat\nu^{(\epsilon)}}
We are now ready to compute derivatives of $a$ with respect to changes in the distribution of environments, as determined by $P$. Complicating the notation is the constraint $\{P:\sum_{e'} P(e,e')=1\text{ for each }e\}$; thus, there can be no sense in speaking of the derivative with respect to changes in $P(e,e')$ for some particular $e, e'$. Instead, we must compute directional derivatives along the direction of some matrix $W$, in the plane $\sum_{e'} W_{e,e'}=0$.  For the purposes of this result we write $a=a(P)$, set $P_{\epsilon}=P+\epsilon W$, and wish to compute the directional derivative $\nabla_{W}a(P)$. One approach is to estimate
$$
\epsilon^{-1}\left( a(P_{\epsilon})-a(P)\right),
$$
and analyze the limit as $\epsilon\to 0$. The perturbations $\epsilon W$ are such that $P_{\epsilon}$ retains the ergodicity and irreducibility of $P$. (The result should be the same whether $\epsilon$ is positive or negative. If $P$ is on the boundary of the set of possible values, one or the other sign may be impossible. Some choices of $W$ may be impossible in both directions.) In the special case in which $W_{e,e'}=1$ and $W_{e,e''}=-1$, with all other entries 0, we are computing the derivative corresponding to a small increase in the rate of transitioning from environment $e$ to $e'$, and a decrease in the frequency of transitioning to $e''$.

We begin by describing separately the \iid case, when the probability $\nu(e)$ that the environment is $e$ is perturbed to $\nu(e) + \epsilon w(e)$, and of course the sum of the $w(e)$ is zero.

\begin{Thm}\label{T:computeSS2}
Suppose the environment process is \iid with distribution $\nu$, and we are given $w\in\R^{K}$ such that $\sum w_{e}=0$. Express $a$ as a function of $\nu$ alone, with the matrices $X_{1},\dots,X_{M}$ assumed fixed. Then
\begin{equation} \label{E:exactSS2}
\nabla_{w} a = \sum_{e\in\M}w_{e}\mathbb{E} \left[ \log(V^T X_{e} U) \right],
\end{equation}
where $U,V^T$ are independent random variables with distributions $\regpi$ and $\tilde{\pi}$ respectively.
This may be approximated by averaging samples of the form
\begin{equation} \label{E:approxSS2}
\sum_{e\in\M}w_{e} \log(V^{(m)T} X_{e} U^{(m)}),
\end{equation}
where
\begin{align*}
U^{(m)}&=X_{e_{1}}\cdots X_{e_{m}}u_{0}\text{ and}\\
V^{(m)T}&=v^T_{0}X_{e'_{m}}\cdots X_{e'_{1}},
\end{align*}
and $e_{1},\dots,e_{m},e'_{1},\dots,e'_{m}$ are independent samples from the distribution $\hat \nu$, and $u_{0}\in\U$ and $v^T_{0}\in\V$.

The systematic error may be bounded uniformly by $2k_{*} r^{m}\|w\|$, while the sampling error at level $p$ on $J$ samples is bounded by
\begin{equation} \label{E:samperrSS2}
2\|w\|\sup_{u\in\U,v^T\in\V} v^T u \left(\frac{-\log p}{2J}\right)^{1/2}.
\end{equation}

Suppose the simulated expectation is
$$
\frac{1}{J} \sum_{i=1}^{J} \sum_{e\in\M}w_{e}\mathbb{E} \left[ \log\|V^{(m)T}(j) X_{e} U^{(m)}(j)\| \right],
$$
where
\begin{align*}
U^{(m)}(j)&=X_{e_{1}(j)}\cdots X_{e_{m}(j)}u_{0}=: Y_{m}(j) u_{0}\text{ and}\\
V^{(m)T}(j)&=v^T_{0}X_{e'_{m}(j)}\cdots X_{e'_{1}(j)}=:v^T_{0}Y'_{m}(j).
\end{align*}
We may also bound the systematic error by
\begin{equation} \label{E:syserrSS2}
\begin{split}
\frac{\|w\|}{J} \sum_{i=1}^{J} &\left( \sup_{u,u'\in \U} \rho(Y_{m}(j) u,Y_{m}(j) u') +  \hspace*{-3mm}\sup_{v^T,v^{'T}\in \U} \hspace*{-3mm}\rho(v^T Y_{m}(j),v'{}^{T} Y_{m}(j) ) \right)\\
&\le \frac{\|w\|}{J} \sum_{i=1}^{J} \left( \Delta\bigl(Y_{m}(j)\bigr) + \Delta\bigl(Y_{m}(j)^{T}\bigr)\right)
\end{split}
\end{equation}
\end{Thm}

The preceding result follows as a special case of the more general result for Markov environments.

\begin{Thm}\label{T:computeSS3}
Suppose the environment process is a Markov chain with transition matrix $P$, with each $e_{i}$ having nonzero probability in the stationary distribution $\nu$. Suppose we have a smooth curve of stochastic matrices $P^{(x)}$, with $P^{(0)}=P$, and where the parameter $x$ takes values either in a two-sided interval $[-\epsilon_{0},\epsilon_{0}]$, or a one-sided interval $[0,\epsilon_{0}]$. Let $W=\partial P^{(x)}/\partial x$, an $M \times  M$ matrix whose rows all sum to 0. The matrices $X_{1},\dots,X_{M}$ are assumed fixed. Then
\begin{equation} \label{E:exactSS3}
a'(0) = \sum_{\te,e\in\M}\nu_{\te}W_{\te,e}\left(\zeta_{\te}+\mathbb{E} \left[ \log\frac{V^T_{e}X_{e}X_{\te} U_{\te}}{\|V^{T}_{e}X_{e}\|} \right]\right),
\end{equation}
where $U_{\te},V^T_{e}$ are independent random variables with distributions $\regpi_{\te}$ and $\tilde{\pi}_{e}$ respectively.

The quantities $\zeta_{\te}$ may be approximated, with error bounds, according to the algorithm described in section \ref{sec:marcoup}.

The other part of the expression may be approximated by averaging samples of the form
\begin{equation} \label{E:approxSS3}
 \sum_{\te,e\in\M}\nu_{\te}W_{\te,e} \log\frac{V^{(m)T} X_{\te} X_{e} U^{(m)}}{\|V^{(m)} X_{e} \|},
\end{equation}
where
\begin{align*}
U^{(m)}&=X_{\te_{0}}\cdots X_{\te_{m}}u_{0}\text{ and}\\
V^{(m)T}&=v^T_{0}X_{e_{m}}\cdots X_{e_{0}},
\end{align*}
and $\te=\te_{0},\te_{1},\dots,e_{m}$ is a Markov chain with transition matrix $\tP$, and $e=e_{0},e_{1},\dots,e_{m}$ is an independent  Markov chain with transition probabilities $P$, and $u_{0}\in\U$ and $v^T_{0}\in\V$.

The systematic error may be bounded uniformly by $2k_{2} r^{m}\|(\nu |W|)\|$, while the sampling error at level $p$ on $J$ samples bounded by
\begin{equation} \label{E:samperrSS3}
2\bigl\|\nu^{T}|W| \bigr\| \sup_{u\in\U,v^T \in\V} v^T u \left(\frac{-\log p}{2J}\right)^{1/2}.
\end{equation}

Suppose the simulated expectation is
$$
\rec{J}\sum_{j=1}^{J} \sum_{\te,e=\in\M}\nu_{\te}W_{\te,e} \log\frac{V^{(m,e)T}(j) X_{e}X_{\te}U^{(m,\te)}(j)}{\|V^{(m,e)T}(j) X_{e}\|},
$$
where
\begin{align*}
U^{(m,\te)}(j)&=\frac{X_{e_{1}(j)}\cdots X_{e_{m}(j)}u_{0}}{\|X_{e_{1}(j)}\cdots X_{e_{m}(j)}u_{0}\|}=: \frac{\tY_{m,\te}(j) u_{0}}{\|\tY_{m,\te}(j) u_{0}\|}\text{ and}\\
V^{(m,e)T}(j)&=\frac{v^T_{0}X_{e_{m}(j)}\cdots X_{e_{1}(j)}}{\|v^T_{0}X_{e_{m}(j)}\cdots X_{e_{1}(j)}\|}=:\frac{v^T_{0}Y_{m,e}(j)}{\|v^T_{0}Y_{m,e}(j)\|}.
\end{align*}
We may also bound the systematic error by
\begin{equation} \label{E:syserrSS3}
\begin{split}
\frac{1}{J} \sum_{j=1}^{J} &\sum_{\te,e\in\M}\nu_{\te}|W_{\te,e}|\left( \sup_{u,u'\in \U} \rho(\tY_{m,\te}(j) u,\tY_{m,\te}(j) u') +  \sup_{v^T,v^{'T}\in \V} \rho(v^T Y_{m,e}(j)X_{e},v'{}^{T} Y_{m,e}(j) X_{e}) \right)\\
&\le \frac{1}{J}\sum_{j=1}^{J} \sum_{\te,e\in\M}\nu_{\te}|W_{e,\te}| \left( \Delta\bigl(\tY_{m,\te}(j)\bigr) + \Delta\bigl(Y_{m,e}(j)^{T} X_{e}\bigr)\right).
\end{split}
\end{equation}
\end{Thm}

We note that expressions like \eqref{E:exactSS2} and \eqref{E:exactSS3} are examples of what \cite{pB92} calls ``ersatz derivatives''. In a rather different class of applications Br\'emaud suggests applying maximal coupling.

\nc{\sys}{\on{Error}_{\on{sys}}}

\section{Discussion}
Our results provide analytical formulas and simulation estimators for the derivatives of stochastic growth rate with respect to the transition probability matrix or the population projection matrices. We have concentrated here on the theoretical results; although this may not be obvious, we have made considerable effort at brevity. Partly for this reason, we will present elsewhere numerical applications of these results. We expect that our results should carry over to integral population models (IPMs), given the strong parallels between the stochastic ergodicity properties of IPMs and matrix models \citep{ellner2007sto}.

Our results apply not only to stochastic structured populations but to any stochastic system in which a Lyapunov exponent of a product of random matrices determines stability or other dynamic properties. Examples include the net reproductive rate in epidemic models and some models of network dynamics. An obvious application of our results is to the analysis of optimal life histories, i.e., environment-to-projection matrix maps that maximize the stochastic growth rate. As discussed by \cite{mcnamara1997opt}, this optimization problem translates into what is called an average reward problem in stochastic control theory, and so our results may be more generally useful in such control problems.

\section{Acknowledgements}
We thank NIA (BSR) for support under 1P01 AG22500. David Steinsaltz was supported by a New Dynamics of Ageing grant, a joint interdisciplinary program of the UK research councils.

\newpage
\section*{Appendix}
\setcounter{section}{1}
\renewcommand{\thesection}{\Alph{section}}

Proofs of the theorems.

\subsection{Estimating the stochastic growth rate} \label{sec:mc}
We prove here Theorem \ref{T:computeLE}. The quantity we are trying to compute is
\begin{equation} \label{E:computeLE}
a=\Ex\left[\log \|X_{e}U\|\right],
\end{equation}
where $(U,e)$ is selected from the distribution $\regpi$. Let $e_{0},e_{1},e_{2},\dots$ be a realization of the stationary Markov chain with transition matrix $P$. Let $Y_{m}:=X_{e_{m}}X_{e_{m-1}}\cdots X_{e_{1}}$. Let $u_{0}\in \U$ be chosen, and let $U$ be a random variable with distribution $\regpi_{e_{0}}$. Then $a=\Ex [\log\|X_{e_{m+1}} Y_{m}U\|/\|Y_{m}U\|]$, which may be approximated by $\Ex [\log\|X_{e_{m+1}} Y_{m}u_{0}\|/\|Y_{m}u_{0}\|]$.

If we identify systematic error with bias, this is
$$
\sys=\left|\Ex\left[\log \frac{\|X_{e_{m+1}}Y_{m} u_{0}\|}{\|Y_{m} u_{0}\|}\right]-\Ex \left[ \log \frac{\|X_{e_{m+1}}Y_{m} U\|}{\|Y_{m} U\|}\right]\right|,
$$
since $(Y_{m}U/\|Y_{m}U\|,e_{m})$ also has the distribution $\regpi$ (if $Y_{m}$ and $U$ are taken to be independent). Thus
\begin{align*}
\sys&\le \Ex\left[\left|\log \frac{\|X_{e_{m+1}}Y_{m} u_{0}\|}{\|Y_{m} u_{0}\|}-\log \frac{\|X_{e_{m+1}}Y_{m} U\|}{\|Y_{m} U\|}\right|\right]\\
&\le \Ex\left[\sup_{u,u'\in\U}\left|\log \frac{\|X_{e_{m+1}}Y_{m} u\|}{\|Y_{m} u\|}-\log \frac{\|X_{e_{m+1}}Y_{m} u'\|}{\|Y_{m} u'\|}\right|\right]\\
&\le k_{2}r^{m},
\end{align*}
by \eqref{E:rhobound}. The corresponding bound on the sampling error may be computed from \eqref{E:hoeffding2}.

For a particular choice of of $e_{1},\dots,e_{m+1}$ and $U$ we can also represent the random systematic error as
$$
\left|\log \frac{\|X_{e_{m+1}}Y_{m} u_{0}\|}{\|Y_{m} u_{0}\|}-\log \frac{\|X_{e_{m+1}}Y_{m} U\|}{\|Y_{m} U\|}\right|,
$$
which may be bounded by the summand in \eqref{E:syserrLE}.

\subsection{Estimating sensitivities: Matrix entries} \label{sec:matent}
We prove here Theorem \ref{T:computeSS1}. As discussed at the end of section \ref{sec:convergence}, we may assume that the compact sets  $\U$ and $\V$ are stable and satisfy the bounds of section \ref{sec:convergence} simultaneously for all $X_{e'}^{(\epsilon)}$. The stationary distributions corresponding to products of the perturbed matrices are denoted $\regpi^{(\epsilon)}$ and $\tilde{\pi}^{(\epsilon)}_{e'}$, and the corresponding regular conditional distributions are $\regpi^{(\epsilon)}_{e}$ and $\tilde{\pi}^{(\epsilon)}_{e'}$

The derivative $a'(0)$ may be written as
\begin{align*}
\lim_{\epsilon\to 0} \epsilon^{-1}&\left(\lim_{m\to\infty}\Ex\left[\log \frac{\|X_{e_{m}}^{(\epsilon)}X_{e_{m-1}}^{(\epsilon)}\cdots X_{e_{1}}^{(\epsilon)}u_{0}\|}{\|X_{e_{m-1}}^{(\epsilon)}\cdots X_{e_{1}}^{(\epsilon)}u_{0}\|}\right]-\lim_{m\to\infty}\Ex\left[\log \frac{\|X_{e_{m}}X_{e_{m-1}}\cdots X_{e_{1}}u_{0}\|}{\|X_{e_{m-1}}\cdots X_{e_{1}}u_{0}\|} \right]\right)\\
&=\lim_{\epsilon\to 0} \lim_{m\to\infty}\sum_{s=1}^{m-1}\Ex\biggl[\epsilon^{-1}\Bigl(\log \frac{\|X_{e_{m}}^{(\epsilon)}X_{e_{m-1}}^{(\epsilon)}\cdots X_{e_{s}}^{(\epsilon)}X_{e_{s-1}}\cdots X_{e_{1}}u_{0}\|}{\|X_{e_{m-1}}^{(\epsilon)}\cdots X_{e_{s}}^{(\epsilon)}X_{e_{s-1}}\cdots X_{e_{1}}u_{0}\|}\\
&\hspace*{5cm}-\log \frac{\|X_{e_{m}}^{(\epsilon)}X_{e_{m-1}}^{(\epsilon)}\cdots X_{e_{s+1}}^{(\epsilon)}X_{e_{s}}\cdots X_{e_{0}}u_{0}\|}{\|X_{e_{m-1}}^{(\epsilon)}\cdots X_{e_{s+1}}^{(\epsilon)}X_{e_{s}}\cdots X_{e_{1}}u_{0}\|}\Bigr)\biggr].
\end{align*}
where $e_{0},e_{1},\dots$ is a realization of the stationary Markov chain with transition probabilities $P$.
Define $a_{s,m}(\epsilon)$ to be the summand on the right-hand side above. By \eqref{E:rhobound},
$$
|a_{s,m}(\epsilon)|\le \epsilon^{-1}k_{2} r^{s-1} \sup_{u\in\U}\rho(X_{e}^{(\epsilon)}u,X_{e}u)\le  C r^{s-1},
$$
where
$$
C=2k_{2}\sup_{u\in\U} \max_{e\in\M} \max_{1\le \ell\le K}\frac{|\tX_{e} u|_{\ell}}{(X_{e}u)_{\ell}}.
$$

Since $\sup_{\epsilon}\sum_{s=m_{0}}^{\infty} \bigl| a_{s,m}(\epsilon)\bigr|\to 0$ as $m_{0}\to\infty$, we may exchange the order of the limits, to see that
\begin{equation} \label{E:reversed}
\begin{split}
a'(0)&= \hspace*{-2mm}\lim_{m\to\infty} \lim_{\epsilon\to 0} \epsilon^{-1}\sum_{s=1}^{m}\biggl(\Ex\biggl[\log \frac{\|X_{e_{m}}^{(\epsilon)}X_{e_{m-1}}^{(\epsilon)}\cdots X_{e_{s+1}}^{(\epsilon)}X_{e_{s}}\cdots X_{e_{0}}u_{0}\|}{\|X_{e_{m-1}}^{(\epsilon)}\cdots X_{e_{s+1}}^{(\epsilon)}X_{e_{s}}\cdots X_{e_{0}}u_{0}\|}\biggr]\\
&\hspace*{2cm}-\Ex\biggl[\log \frac{\|X_{e_{m}}^{(\epsilon)}X_{e_{m-1}}^{(\epsilon)}\cdots X_{e_{s+1}}^{(\epsilon)}X_{e_{s}}\cdots X_{e_{0}}u_{0}\|}{\|X_{e_{m-1}}^{(\epsilon)}\cdots X_{e_{s}}^{(\epsilon)}X_{e_{s-1}}\cdots X_{e_{0}}u_{0}\|}\biggr] \biggl).
\end{split}
\end{equation}
This limit is the same for any choice of $u_{0}$, hence would also be the same if we replaced $u_{0}$ by a random $U$, with any distribution on $\U$. We choose $U$ to have the distribution $\regpi_{e_{0}}$, independent of the rest of the Markov chain. By the invariance property of the distributions $\regpi$,
\begin{equation} \label{E:reversed2}
\begin{split}
a'(0)&= \hspace*{-2mm}\lim_{m\to\infty} \lim_{\epsilon\to 0} \epsilon^{-1}\biggl(\sum_{s=1}^{m} \Ex\left[\log \frac{\|X_{e_{m}}^{(\epsilon)}X_{e_{m-1}}^{(\epsilon)}\cdots X_{e_{s}}^{(\epsilon)}X_{e_{s-1}}\cdots X_{e_{0}}U\|}{\|X_{e_{m-1}}^{(\epsilon)}\cdots X_{e_{s}}^{(\epsilon)}X_{e_{s-1}}\cdots X_{e_{0}}U\|}\right]\\
&\hspace*{2.5cm} -\Ex\left[\log \frac{\|X_{e_{m}}^{(\epsilon)}X_{e_{m-1}}^{(\epsilon)}\cdots X_{e_{s+1}}^{(\epsilon)}X_{e_{s}}\cdots X_{e_{0}}U\|}{\|X_{e_{m-1}}^{(\epsilon)}\cdots X_{e_{s+1}}^{(\epsilon)}X_{e_{s}}\cdots X_{e_{0}}U\|}\right]\biggr)\\
&= \lim_{m\to\infty} \sum_{s=0}^{m} \lim_{\epsilon\to 0} \epsilon^{-1}\Ex\biggl[\log \frac{\|X_{e_{m}}^{(\epsilon)}X_{e_{m-1}}^{(\epsilon)}\cdots X_{e_{s}}^{(\epsilon)}U_{s-1}\|}{\|X_{e_{m-1}}^{(\epsilon)}\cdots X_{e_{s}}^{(\epsilon)}U_{s-1}\|}\\
&\hspace*{4cm} -\log \frac{\|X_{e_{m}}^{(\epsilon)}X_{e_{m-1}}^{(\epsilon)}\cdots X_{e_{s+1}}^{(\epsilon)}X_{e_{s}}U_{s-1}\|}{\|X_{e_{m-1}}^{(\epsilon)}\cdots X_{e_{s+1}}^{(\epsilon)}X_{e_{s}}U_{s-1}\|}\biggr],
\end{split}
\end{equation}
where $(U_{s-1},e_{s})$ has distribution $\regpi$.

For $m\ge s\ge 1$ define functions
$$
f_{s}(\epsilon,\delta):=\frac{\indic^{T} X_{e_{m}}^{(\delta)}X_{e_{m-1}}^{(\delta)}\cdots X_{e_{s+1}}^{(\delta)}X_{e_{s}}^{(\epsilon)}U_{s-1}}{\indic^{T} X_{e_{m-1}}^{(\delta)}\cdots X_{e_{s+1}}^{(\delta)} X_{e_{s}}^{(\epsilon)}U_{s-1}},
$$
where the denominator is understood to be 1 for $s=m$.
Take $f_{0}(\epsilon,\delta):= \|X_{e_{0}}^{(\epsilon)}U\|$. The summand on the right of \eqref{E:reversed2} may be written as
\begin{equation} \label{E:fdiff}
\lim_{\epsilon\to 0} \epsilon^{-1} \Ex \left[ \log f_{s}(\epsilon,\epsilon) - \log f_{s}(0,\epsilon) \right].
\end{equation}
By \eqref{E:rhobound},
$$
\epsilon^{-1}\left|\log f_{s}(\epsilon,\delta) - \log f_{s}(0,\delta)\right| \le k_{2}r^{s} \epsilon^{-1}\rho(X_{e_{s}}^{(\epsilon)} U,X_{e_{s}} U),
$$
which is bounded for $\epsilon$ in a neighborhood of 0, so the Bounded Convergence Theorem turns \eqref{E:fdiff} into
\begin{equation} \label{E:fdiff2}
\Ex \left[ \lim_{\epsilon\to 0} \epsilon^{-1} \bigl(\log f_{s}(\epsilon,\epsilon) - \log f_{s}(0,\epsilon) \bigr)\right]=
\Ex \left[ \frac{\partial \log f_{s}}{\partial \epsilon} (0,0) \right]\\
\end{equation}
since for any choice of $e_{0},\dots,e_{s}$ and $U$ the function $f_{s}$ is continuously differentiable at $(0,0)$ and bounded away from 0. We have, by linearity of the matrix product and $\|\cdot\|$,
\begin{align*}
&\frac{\partial \log f_{s}}{\partial \epsilon} (0,0) = \frac{\indic^{T}X_{e_{m}}\cdots X_{e_{s+1}}\frac{\partial}{\partial\epsilon} X_{e_{s}}^{(\epsilon)} U_{s-1} }{\|X_{e_{m}}\cdots X_{e_{s}}U_{s-1}\|}\\
&\hspace*{4cm}-\frac{\indic^{T} X_{e_{m-1}}\cdots X_{e_{s+1}}\frac{\partial}{\partial\epsilon} X_{e_{s}}^{(\epsilon)} U_{s-1} }{\|X_{e_{m-1}}\cdots X_{e_{s}}U_{s-1}\|}\\
&\qquad =\begin{cases}
\left( \frac{\indic^{T}X_{e_{m}}\cdots X_{e_{s+1}}\tX_{e_{s}} U_{s-1}}{\|X_{e_{m}}\cdots X_{e_{s}}U_{s-1}\|}-\frac{\indic^{T} X_{e_{m-1}}\cdots X_{e_{s+1}}\tX_{e_{s}} U_{s-1}}{\|X_{e_{m-1}}\cdots X_{e_{s}}U_{s-1}\|}\right)& \text{for } 1\le s\le m-1,\\
\frac{\indic^{T}\tX_{e_{m}} U_{m-1}}{\|X_{e_{m}} U_{m-1}\|}  & \text{for } s=m.
\end{cases}
\end{align*}

Combining this with \eqref{E:reversed2} yields the telescoping sum
\begin{align*}
a'(0)&= \lim_{m\to\infty} \biggl(\sum_{t=1}^{m} \Ex \left[\frac{\indic^{T}X_{e_{t}}\cdots X_{e_{1}}\tX_{e_{0}} U}{\indic^{T}X_{e_{t}}\cdots x_{e_{1}}X_{e_{0}}U} \right] \\
&\hspace*{3cm}- \sum_{t=1}^{m-1} \Ex \left[ \frac{\indic^{T}X_{e_{t}}\cdots X_{e_{1}}\tX_{e_{0}} U}{\indic^{T}X_{e_{t}}\cdots x_{e_{1}}X_{e_{0}}U} \right] \biggr)\\
&= \lim_{m\to\infty} \Ex \left[ \frac{\indic^{T}X_{e_{m}}\cdots X_{e_{1}}\tX_{e_{0}} U}{\indic^{T}X_{e_{m}}\cdots X_{e_{1}}X_{e_{0}} U} \right],
\end{align*}
where in the last line $(U,e_{0})$ has the distribution $\pi$. Define $V^T_{m}:=\indic^{T}X_{e_{m}}\cdots X_{e_{1}}/\|\indic^{T}X_{e_{m}}\cdots X_{e_{1}}\|$. Then $V^T_{m}$ converges in distribution to $V$, with distribution $\tilde{\pi}_{e_{0}}$, and so
$$
a'(0)=\lim_{m\to\infty} \Ex \left[ \frac{V^T_{m}\tX_{e} U}{V^T_{m}X_{e}U} \right]=\sum_{e\in\M}\nu_{e} \Ex \left[ \frac{V^{T}_{e}\tX_{e} U_{e}}{V^{T}_{e}X_{e}U_{e}} \right],
$$
which is identical to \eqref{E:exactSS1}.

Now we estimate the error. We use the representation 
$$
U:=\frac{X_{\tilde{e}_{1}}\cdots X_{\tilde{e}_{m}}U_{0}}{\|X_{\tilde{e}_{1}}\cdots X_{\tilde{e}_{m}}U_{0}\|}\text{ and } 
V^{T}:=\frac{V^T_{0}X_{e_{m}}\cdots X_{e_{1}}}{\|V^T_{0}X_{e_{m}}\cdots X_{e_{1}}\|}, 
$$
where $U_{0}$ and $V^T_{0}$ are assumed to have distributions $\pi_{\te_{m}}$ and $\tilde{\pi}_{e_{m}}$ respectively. Then
\begin{equation} \label{E:logss1est}
\begin{split}
\Bigl|\frac{V^{(m)T}\tX_{e} U^{(m)}}{V^{(m)T} X_{e} U^{(m)}} - \log \frac{V^T \tX_{e} U}{V^T X_{e} U}\Bigr| &\le 2\left( e^{2\rho(V,V^{(m)T}) +2\rho(U,U^{(m)})}-1\right) \frac{V \tX_{e} U}{VX_{e} U}\\
&\le 2 \left(e^{4k_{2}r^{m}}-1\right) \frac{V \tX_{e} U}{VX_{e} U},
\end{split}
\end{equation}
by \eqref{E:rhobound}. This implies the uniform bound on systematic error, and the bound on sampling error \eqref{E:samperrSS1} follows from applying \eqref{E:hoeffding2} to a trivial bound on the terms in the average. The simulated bound \eqref{E:syserrSS1} also follows directly from \eqref{E:logss1est}.

\nc{\Pep}{P^{(\epsilon)}}
\nc{\eep}{e^{(\epsilon)}}
\nc{\nep}{\nu^{(\epsilon)}}
\nc{\tPep}{\tP^{(\epsilon)}}
\subsection{Estimating sensitivities: Markov environments} \label{sec:markov}
We prove here Theorem \ref{T:computeSS3} by a combination of the coupling method and importance sampling. We use importance sampling for the actual computation, but coupling provides a more direct path to validating the crucial exchange of limits. Suppose we are given any $\epsilon$ such that $P+\epsilon W$ and $P-\epsilon W$ are both stochastic matrices.

Given two distributions $q$ and $q'$ on $\{1,\dots,M\}$, we  define a {\em standard coupling} between $q$ and $q'$. Suppose we are given  a uniform random variable $\omega$ on $[0,1]$. Let $\M_{-}:=\{e: q_{e}<q'_{e}\}$ and $\M_{+}:=\{e: q'_{e}<q_{e}\}$. Let $\delta:=\sum_{e\in\M_{-}} (q'_{e}-q_{e})=\sum_{e\in\M_{+}} (q_{e}-q'_{e})$. We define three random variables $\te$ on $\M$, $\te_{+}$ on $\M_{+}$, and $\te_{-}$ on $\M_{-}$, according to the following distributions:
\begin{align*}
\P\{\te=e\}&=\min\{q_{e},q'_{e}\}/(1-\delta),\\
\P\{\te_{+}=e\}&=(q_{e}-q'_{e})_{+}/\delta,\\
\P\{\te_{-}=e\}&=(q'_{e}-q_{e})_{+}/\delta.
\end{align*}
The joint distribution is irrelevant, but for definiteness we let them be independent. Then we define the coupled pair $(e,e')$ by
\begin{equation} \label{E:standardcouple}
\begin{split}
(\te,\te)&\text{ if } \omega>\delta;\\
(\te_{+},\te_{-})&\text{ if } \omega\le\delta.
\end{split}
\end{equation}
Then $e$ has distribution $q$, $e'$ has distribution $q'$, and $e=e'$ with probability $1-\delta$. This $\delta$ is called the {\em total-variation distance} between $q$ and $q'$.

We write $\Ex_{P}$ for the expectation with respect to the distribution that makes $e_{0},\dots,e_{m}$ a stationary Markov chain with transition matrix $P$. Define $\nep$ to be the stationary distribution corresponding to $\Pep$, and define $\tPep$ to be the time-reversed chain of $\Pep$. We define
$$
g(m,\epsilon;u):= \Ex_{\Pep}\left[\log\frac{\|X_{e_{m}}X_{e_{m-1}}\cdots X_{e_{0}} u\|}{\|X_{e_{m-1}}\cdots X_{e_{0}} u\|}\right]- \Ex_{P}\left[\log\frac{\|X_{e_{m}}X_{e_{m-1}}\cdots X_{e_{0}} u\|}{\|X_{e_{m-1}}\cdots X_{e_{0}} u\|}\right] .
$$
By the time-reversal property,
$$
g(m,\epsilon;u)= \Ex_{\tPep}\left[\log\frac{\|X_{e_{0}}X_{e_{1}}\cdots X_{e_{m}} u\|}{\|X_{e_{1}}\cdots X_{e_{m}} u\|}\right]- \Ex_{\tP}\left[\log\frac{\|X_{e_{0}}X_{e_{1}}\cdots X_{e_{m}} u\|}{\|X_{e_{1}}\cdots X_{e_{m}} u\|}\right] .
$$

For $\epsilon>0$ we couple a sequence $e_{0},\dots,e_{m}$ selected from the distribution $\tP$ to a sequence $\eep_{0},\dots,\eep_{m}$ selected from the distribution $\tPep$ as follows: We start by choosing $(e_{0},\eep_{0})$ according to the standard coupling of $(\nu,\nep)$. Assume now that we have produced sequences of length $i$, ending in $e_{i-1}$ and $\eep_{i-1}$. We then produce $(e_{i},\eep_{i})$ according to the standard coupling of row $e_{i-1}$ of $\tP$ to row $\eep_{i-1}$ of $\tPep$. (To simplify the typography in some places, we  use $e(i)$ and $\eep(i)$ interchangeably with $e_{i}$ and $\eep_{i}$.)

Let $\delta=\delta(\epsilon)$ be the maximum of the total variation distance between $\nu$ and $\nep$, and all of the pairs of rows. It is easy to see that there is a constant $c$ such that $\delta\le c\epsilon$ for $\epsilon$ sufficiently small. Define $\omega_{1},\omega_{2},\dots$ to be an i.i.d.\ sequence of uniform random variables on $[0,1]$, and two sequences of random times as follows: $T_{0}:=S_{0}:=-1$, and
\begin{align*}
T_{i+1}&=\min\bigl\{ t> S_{i}\, : \, \omega_{t}\le \delta \bigr\},\\
S_{i+1}&=\min\bigl\{ t> T_{i+1}\, : \, \eep_{t}= e_{t} \bigr\}.
\end{align*}
Thus, $\eep_{t}=e_{t}$ for all $S_{i}\le t< T_{i+1}$. Define for any $u_{0}\in \U$ the random vector
$$
U_{t}:=\lim_{m\to\infty} \frac{X_{e(t)}\cdots X_{e(t+m)} u_{0}}{\|X_{e(t)}\cdots X_{e(t+m)} u_{0}\|},
$$
and define a version of $g$ conditioned on $T_{1}$ and $T_{2}$
\begin{align*}
g(m,\epsilon;u;T_{1},T_{2})&:= \Ex_{\tPep}\left[\log\frac{\|X_{e_{0}}X_{e_{1}}\cdots X_{e_{m}} u\|}{\|X_{e_{1}}\cdots X_{e_{m}} u\|}\, \Bigl| \, T_{1},T_{2}\right]\\
&\hspace*{2cm}- \Ex_{\tP}\left[\log\frac{\|X_{e_{0}}X_{e_{1}}\cdots X_{e_{m}} u\|}{\|X_{e_{1}}\cdots X_{e_{m}} u\|}\, \Bigl| \, T_{1},T_{2}\right] .
\end{align*}
Then for any $u\in\U$,
\begin{equation} \label{E:nablaW}
\nabla_{W} a(P)=a'(0)=\lim_{\epsilon\downarrow 0} \lim_{m\to\infty}\epsilon^{-1}\Ex \left[g(m,\epsilon;u;T_{1},T_{2})\right].
\end{equation}
We also define
\begin{align*}
\gamma(\epsilon;T_{1},T_{2}):=\Ex&\Bigl[\log\frac{\|X_{\eep(0)}\cdots X_{\eep(S_{1}-1)} U_{S_{1}}\|}{\|X_{\eep(1)}\cdots X_{\eep(S_{1}-1)} U_{S_{1}}\|}\\
&\hspace*{15mm}- \log \frac{\|X_{e(0)}X_{e(1)} \cdots X_{e(S_{1}-1)} U_{S_{1}}\|}{\|X_{e(1)}\cdots  X_{e(S_{1}-1)} U_{S_{1}}\|}\, \Bigl| \, T_{1},T_{2}\Bigr]
\end{align*}

We break up these expectations into their portion overlapping three different events:
\begin{enumerate}
\item $\{m< T_{1}\}$;
\item $\{T_{2}\ge m>T_{1}\}$;
\item $\{m\ge T_{2}\}$.
\end{enumerate}

\bigskip
\noindent
{\bf On the event $\{m< T_{1}\}$} we have $g(m,\epsilon;u;T_{1},T_{2})=0$, and $T_{1}-m$ is geometrically distributed with parameter $\delta$. By \eqref{E:rhobound}, $\gamma$ is bounded by $k_{2}r^{S_{1}-1}\le k_{2}r^{T_{1}-1}$, meaning that
\begin{equation} \label{E:firstgmbound}
\begin{split}
\Ex \left[\bigl| \gamma(\epsilon;T_{1},T_{2}) - g(m,\epsilon;u;T_{1},T_{2}) \bigr|\, \Bigl| \, T_{1}>m\right]&=\Ex \left[\bigl| \gamma(\epsilon;T_{1},T_{2}) \bigr|\, \Bigl| \, T_{1}>m \right]\\
&\le k_{2} \Ex\left[r^{T_{1}-1}\cond T_{1}>m \right]\\
&\le\frac{k_{2}}{1-r} r^{m-1} \delta.
\end{split}
\end{equation}

\bigskip
\noindent
{\bf On the event $\{T_{2}>m\ge T_{1}\}$}: We have $\eep(i)=e(i)$ for $i<T_{1}$ and for $S_{1}\le i\le m$.  Thus, if $S_{1}\le m$,
\begin{align*}
U_{S_{1}}&= X_{e(S_{1})}\cdots X_{e(m)} U_{m+1}/\|X_{e(S_{1})}\cdots X_{e(m)} U_{m+1}\|\\
&=X_{\eep(S_{1})}\cdots X_{\eep(m)} U_{m+1}/\|X_{\eep(S_{1})}\cdots X_{\eep(m)} U_{m+1}\|\\
\end{align*}
Thus we may write
\begin{equation} \label{E:secondgmbound}
\begin{split}
\bigl| &\gamma(\epsilon;T_{1},T_{2}) - g(m,\epsilon;u;T_{1},T_{2}) \bigr| \\
&\le \Biggl| \Ex\left[\log\frac{\|X_{e(0)}\cdots X_{e(T_{1}-1)} X_{\eep(T_{1})}\cdots X_{\eep(m)}U'\|}{\|X_{e(1)}\cdots X_{e(T_{1}-1)} X_{\eep(T_{1})}\cdots X_{\eep(m)}U'\|}\biggl| T_{1},T_{2}\right] \\
&\hspace*{2.5cm}- \Ex\left[\log\frac{\|X_{e(0)}\cdots X_{e(T_{1}-1)} X_{\eep(T_{1})}\cdots X_{\eep(m)}u\|}{\|X_{e(1)}\cdots X_{e(T_{1}-1)} X_{\eep(T_{1})}\cdots X_{\eep(m)}u\|}\biggl| T_{1},T_{2}\right] \Biggr|\\
&\hspace*{.5cm}+ \Biggl|\Ex\Biggl[\log \frac{\|X_{e(0)}X_{e(1)} \cdots X_{e(m)}U''\|}{\|X_{e(1)}\cdots X_{e(m)}U''\|}\\
&\hspace*{3cm}-\log \frac{\|X_{e(0)}\cdots X_{e(T_{1}-1)} X_{e(T_{1})}\cdots X_{e(m)}u\|}{\|X_{e(1)}\cdots X_{e(T_{1}-1)} X_{e(T_{1})}\cdots X_{e(m)}u\|}\biggl| T_{1},T_{2}\Biggr]\Biggr|\\
&\le 2k_{2}r^{m},
\end{split}
\end{equation}
where
\begin{align*}
U'&=\begin{cases}
U_{m+1}&\text{if } S_{1}\le m,\\
\frac{X_{\eep(m+1)} \cdots X_{\eep(S_{1}-1)} U_{S_{1}}}{\|X_{\eep(m+1)} \cdots X_{\eep(S_{1}-1)} U_{S_{1}}\|}&\text{if } S_{1}> m.
\end{cases}\\
\text{and } U''&=\begin{cases}
U_{m+1}&\text{if } S_{1}\le m,\\
\frac{X_{(m+1)} \cdots X_{(S_{1}-1)} U_{S_{1}}}{\|X_{(m+1)} \cdots X_{(S_{1}-1)} U_{S_{1}}\|}&\text{if } S_{1}> m.
\end{cases}
\end{align*}

\bigskip
\noindent
{\bf On the event $\{T_{2}\le m\}$}: The above approach shows that
\begin{equation} \label{E:thirdgmbound}
\bigl| \gamma(\epsilon;T_{1},T_{2}) - g(m,\epsilon;u;T_{1},T_{2}) \bigr| \le 2k_{2} r^{T_{2}-1}.
\end{equation}

\bigskip
\noindent
Combining \eqref{E:firstgmbound}, \eqref{E:secondgmbound} and \eqref{E:thirdgmbound}, we obtain
\begin{equation} \label{E:fourthgmbound}
\begin{split}
\bigl| \gamma(\epsilon;T_{1},T_{2}) -& g(m,\epsilon;u;T_{1},T_{2}) \bigr| \\
&\le k_{2}r^{T_{1}-1} \indic_{\{T_{1}>m\}} + 2k_{2} r^{m}\indic_{\{T_{1}\le m\}} + 2k_{2} r^{T_{2}-1}.
\end{split}
\end{equation}
Taking the expectation with respect to the distribution of $T_{1}$ and $T_{2}$, using the fact that $T_{1}$ and $T_{2}-S_{1}$ are independent with distribution geometric with parameter $\delta$, we obtain
\begin{equation} \label{E:finalgmbound}
\begin{split}
\Ex\Bigl[ \bigl| \gamma(\epsilon;T_{1},T_{2}) -&g(m,\epsilon;u;T_{1},T_{2}) \bigr| \Bigr] \\
&\le \frac{k_{2}}{1-r} r^{m-1} \delta + 2k_{2} \delta m r^{m}+\frac{2k_{2}}{r^{2}(1-r)^{2}} \delta^{2}.
\end{split}
\end{equation}
Since $\delta$ is bounded by a constant times $|\epsilon|$, we may find a constant $C$ such that (by the triangle inequality) for all $\epsilon$, positive integers $m$, and $u\in \U$,
\begin{equation} \label{E:finalgmbound2}
\begin{split}
\Bigl|\Ex\left[ \gamma(\epsilon;T_{1},T_{2})\right] -& \Ex\left[g(m,\epsilon;u;T_{1},T_{2})\right] \Bigr| \\
&\le\Ex\left[ \bigl| \gamma(\epsilon;T_{1},T_{2}) - g(m,\epsilon;u;T_{1},T_{2}) \bigr| \right] \\
&\le C(mr^{m}|\epsilon|+\epsilon^{2}).
\end{split}
\end{equation}

This bound allows us to exchange the limits in \eqref{E:nablaW}:
\begin{equation} \label{E:orderchange}
\begin{split}
a'(0)&= \lim_{\epsilon\to 0} \lim_{m\to\infty} \epsilon^{-1} \Ex\left[g(m,\epsilon;u;T_{1},T_{2})\right]\\
&= \lim_{\epsilon\to 0} \epsilon^{-1} \Ex\left[\gamma|(\epsilon;T_{1},T_{2})\right]\\
&= \lim_{m\to\infty} \lim_{\epsilon\to 0}\epsilon^{-1} \Ex\left[g(m,\epsilon;u;T_{1},T_{2}) \right]\\
&= \lim_{m\to\infty} \frac{d\phantom{\epsilon}}{d\epsilon}\Bigl|_{\epsilon=0}\Ex_{\Pep}\left[\log\frac{\|X_{e_{m}}X_{e_{m-1}}\cdots X_{e_{0}} u\|}{\|X_{e_{m-1}}\cdots X_{e_{0}} u\|}\right]
\end{split}
\end{equation}

Now we apply the method of importance sampling. We may assume without loss of generality that $W(e,e')=0$ whenever $P(e,e')=0$ (using the analyticity of $a$, and the fact that the formula \eqref{E:exactSS3} is nonsingular on the nonnegative orthant).For any function $Z:\M^{m+1}\to \R$,
$$
\Ex_{\Pep}\left[Z(e_{0},\dots,e_{m})\right] = \Ex_{P}\left[Z(e_{0},\dots,e_{m})F(\epsilon;e_{0},\dots,e_{m})\right],
$$
where $F$ is the Radon-Nikodym derivative
\begin{align*}
F(\epsilon;e_{0},\dots,e_{m})&=\frac{d\Pep}{dP}(e_{0},\dots,e_{m})\\
&=\frac{\nep_{e_{0}}}{\nu_{e_{0}}} \prod_{i=0}^{m-1} \frac{\Pep(e_{i},e_{i+1})}{P(e_{i},e_{i+1})}.
\end{align*}
This allows us to rewrite
\begin{equation} \label{E:nablaW2}
\begin{split}
a'(0)
= \lim_{m\to\infty} \frac{d\phantom{\epsilon}}{d\epsilon}\Bigl|_{\epsilon=0}\hspace*{-2mm}\Ex_{P}\biggl[\frac{\nep_{e_{0}}}{\nu_{e_{0}}} \prod_{i=0}^{m-1} \frac{P^{(\epsilon)}(e_{i},e_{i+1})}{P^{(\epsilon)}(e_{i},e_{i+1})}\\
&\times\log\frac{\|X_{e_{m}}X_{e_{m-1}}\cdots X_{e_{0}} u\|}{\|X_{e_{m-1}}\cdots X_{e_{0}} u\|}\biggr]
\end{split}
\end{equation}
For any fixed $m$, there is an upper bound on $\epsilon^{-1}(F(\epsilon;e_{0},\dots,e_{m})-1)$, so we may move the differentiation inside the expectation, to obtain
\begin{equation} \label{E:nablaW3}
\begin{split}
a'(0)&= \lim_{m\to\infty} \Ex_{P}\Biggl[\frac{d\phantom{\epsilon}}{d\epsilon}\Bigl|_{\epsilon=0} \frac{\nep_{e_{0}}}{\nu_{e_{0}}} \prod_{i=0}^{m-1}  \frac{\Pep(e_{i},e_{i+1})}{P(e_{i},e_{i+1})}\\
&\hspace*{4cm}\times\log\frac{\|X_{e_{m}}X_{e_{m-1}}\cdots X_{e_{0}} u\|}{\|X_{e_{m-1}}\cdots X_{e_{0}} u\|}\Biggr]\\
&= \lim_{m\to\infty} \Ex_{P}\left[(\nu_{e_{0}})^{-1}\frac{d\nep_{e_{0}}}{d\epsilon}\Bigl|_{\epsilon=0} \log\frac{\|X_{e_{m}}X_{e_{m-1}}\cdots X_{e_{0}} u\|}{\|X_{e_{m-1}}\cdots X_{e_{0}} u\|}\right]\\
&\hspace*{.5cm}+ \lim_{m\to\infty} \sum_{i=0}^{m-1}\Ex_{P} \left[\frac{W(e_{i},e_{i+1})}{P(e_{i},e_{i+1})}\log\frac{\|X_{e_{m}}X_{e_{m-1}}\cdots X_{e_{0}} u\|}{\|X_{e_{m-1}}\cdots X_{e_{0}} u\|}\right]
\end{split}
\end{equation}

The first limit is 0. To see this, rewrite it as a sum over the possible values of $e_{0}$:
\begin{align*}
\lim_{m\to\infty} \sum_{\te\in\M}\nu_{\te}\Ex_{P}&\left[(\nu_{\te})^{-1}\frac{d\nep_{e}}{d\epsilon}\Bigl|_{\epsilon=0} \log\frac{\|X_{e_{m}}X_{e_{m-1}}\cdots X_{\te} u\|}{\|X_{e_{m-1}}\cdots X_{\te} u\|}\right]\\
&= \lim_{m\to\infty} \sum_{\te\in\M}\frac{d\nep_{\te}}{d\epsilon}\Ex_{P}\left[\log\frac{\|X_{e_{m}}X_{e_{m-1}}\cdots X_{e_{1}}X_{\te} u\|}{\|X_{e_{m-1}}\cdots X_{e_{1}} X_{\te} u\|}\right]
\end{align*}
Since $\nep$ is a probability distribution, it must be that $\sum_{\te=1}^{\M}\frac{d\nep_{\te}}{d\epsilon}=0$. Thus, the expression in the limit becomes 0 if we replace the expectation by a constant, independent of $\te$. By Lemma \ref{L:differentstart} it follows that the limit is 0.

To compute the other limit, we sum over all possible pairs $(e_{i},e_{i+1})=(\te,e)$. The summand becomes
\begin{equation} \label{E:finalsum}
\sum_{\te,e\in\M}\nu(\te)W(\te,e)\Ex \left[\log\frac{\|X_{e_{m}}X_{e_{m-1}}\cdots X_{e_{0}} u\|}{\|X_{e_{m-1}}\cdots X_{e_{0}} u\|}\, \Bigl| \, e_{i}=\te, e_{i+1}=e\right]
\end{equation}
In order to analyze this, we need to consider the distribution of $e_{0},\dots,e_{m}$, conditioned on $e_{i}=\te$ and $e_{i+1}=e$. By the Markov property, this splits into two independent Markov chains: $e=e_{i+1},\dots,e_{m}$ is a Markov chain of length $m-i$, with transition probabilities $P$ and starting point $e$, while $\te=e_{i},e_{i-1},\dots,e_{0}$ is a Markov chain of length $i+1$ with transition probabilities $\tP$ and starting point $\te$. Define two independent infinite sequences $\te_{0},\te_{1},\dots$ and $e_{i+1},e_{i+2},\dots$, which are Markov chains with transitions $\tP$ and $P$ respectively, beginning in $\te_{0}=\te$ and $e_{i+1}=e$. Define for $i\ge 1$, $\tU_{i}(\te):=X_{\te}X_{\te_{1}}\cdots X_{\te_{i}}u$ with $\tU_{0}(\te):=\indic$, and ${\tV}^T_{i}(e):=\indic^{T}X_{e_{i-1}}\cdots X_{e_{1}}X_{e}$  with ${\tV}^T_{0}(e):=\indic^{T}$.
Also define
$$
U_{i}(e):=\frac{{\tU}_{i}(e)}{\|{\tU}_{i}(e)\|}, \qquad
V^T_{i}(e):=\frac{{\tV}^T_{i}(e)}{\|{\tV}^T_{i}(e)\|}.
$$
Since $\|u\|=\indic^{T}u$ for any nonnegative column vector $u$, the expression \eqref{E:nablaW3} becomes
\begin{equation} \label{E:nablaW4}
\begin{split}
a'(0)&=\sum_{\te,e\in\M} \nu(\te) W(\te,e)\lim_{m\to\infty} \biggl\{\sum_{i=0}^{m-1}\biggl(\Ex \left[\log {\tV}^T_{i+1}(e)\tU_{m-i}(\te) \right] \\
&\hspace*{4.5cm}- \Ex \left[ \log {\tV}^T_{i}(e)\tU_{m-i}(\te)\right]\biggr) \biggr\}\\
&=\sum_{\te,e\in\M} \nu(\te) W(\te,e)\lim_{m\to\infty} \biggl\{\sum_{i=0}^{m-1}\biggl(\Ex \left[\log\|{\tV}^T_{i+1}(e)\|\right]\\
&\hspace*{1.5cm}-\Ex \left[\log\|{\tV}^T_{i}(e)\|\right]+\Ex \left[\log V^T_{i+1}(e)U_{i}(\te) \right] \\
&\hspace*{5cm}- \Ex\left[ \log V^T_{i}(e)U_{i}(\te)\right]\biggr) \biggr\}\\
&=\sum_{\te\in\M} \nu(\te) \lim_{m\to\infty} \sum_{e\in\M}W(\te,e)\Ex \left[\log\|{\tV}^T_{m}(e)\|\right]\\
&\hspace*{.5cm}+\sum_{e,\te\in\M} \nu(\te) W(\te,e)\lim_{m\to\infty} \sum_{i=0}^{m-1}\Ex \left[ \log \frac{V^T_{i+1}(e)U_{m-i}(\te)}{ V^T_{i}(e)U_{m-i}(\te)}\right]
\end{split}
\end{equation}
In the last line we have used the fact that $\sum_{e\in\M}W(\te,e)=0$, which means that $\sum_{e\in\M}W(\te,e)\Ex \left[\log\|{\tV}^T_{0}(e)\|\right]=0$ as well, since ${\tV}^T_{0}(e)=\indic^{T}$ is independent of $e$. The same reasoning implies that if we define ${\tV}^T_{i}(\nu)$ to be the version of ${\tV}^T_{i}$ started in the stationary distribution --- for instance, starting from realizations of ${\tV}^T_{i}(e)$, define ${\tV}^T_{i}(\nu)$ to be equal to ${\tV}^T_{i}(e)$ with probability $\nu_{e}$ --- then $\sum_{e\in\M}W(\te,e)\Ex \left[\log\|{\tV}^T_{m}(\nu)\|\right]=0$. The first term on the right-hand side of \eqref{E:nablaW4} may then be written as
\begin{equation} \label{E:zetalim}
\sum_{\te,e\in\M}\nu(\te) W(\te,e)\lim_{m\to\infty} \Ex \left[\log\frac{\|{\tV}^T_{m}(e)\|}{\|{\tV}^T_{m}(\nu)\|}\right]= \sum_{\te,e\in\M} \nu(\te)W(\te,e)\zeta_{e}.
\end{equation}

To compute the second term, we note that $U(\te):=\lim_{i\to\infty} U_{i}(\te)$ exists, with distribution $\regpi_{\te}$, and $\rho(U_{i}(\te),U(\te))\le k_{2}r^{i}$; similarly, $V^T(e):=\lim_{i\to\infty} V^T_{i}(e)$ exists, with distribution $\tilde{\pi}_{e}$, and $\rho(V_{i}(e),V_{i+1}(e))\le k_{2}r^{i}$. Thus
$$
\bigl|\log V^T_{i+1}(e')U_{m-i}(e)-\log V^T_{i}(e')U_{m-i}(e)\bigr|\le \rho\bigl( V_{i}(e'),V_{i+1}(e')\bigr) \le k_{2} r^{i},
$$
We break up the sum on the right-hand side of \eqref{E:nablaW4} into three pieces:
\begin{align*}
\sum_{0\le i\le m-1}\Ex &\left[\log V^T_{i+1}(e)U_{m-i}(\te)-\log V^T_{i}(e)U_{m-i}(\te)\right] \\
&= \sum_{0\le i\le m/2}\Ex \left[ \log V^T_{i+1}(e)U(\te)-\log V^T_{i}(e)U(\te)\right]\\
&\hspace*{1cm}+\sum_{0\le i\le m/2}\Ex \left[ \log \frac{V^T_{i+1}(e)U_{m-i}(\te)}{V^T_{i+1}(e)U(\te)}-\log \frac{V^T_{i}(e)U_{m-i}(\te)}{V^T_{i}(e)U(\te)}\right]\\
&\hspace*{1cm}+\sum_{ m/2<i\le m-1}\Ex \left[ \log V^T_{i+1}(e)U_{m-i}(\te)-\log V^T_{i}(e)U_{m-i}(\te)\right].
\end{align*}
The first sum telescopes to
$$
\Ex \left[ \log V^T_{1+m/2}(e)U(\te)-\log V^T_{0}(e)U(\te)\right]=\Ex \left[ \log V^T_{1+m/2}(e)U(\te)\right],
$$
applying the fact that $V^T_{0}=\indic^{T}$, so that $V^T_{0}(e) U(\te)=\|U(\te)\|=1$.
Applying \eqref{E:ratiobound}, the second and third sums are bounded by
$$
\sum_{0\le i\le m/2}2k_{2} r^{m-i}+\sum_{ m/2<i\le m-1}k_{2}r^{i}\le \frac{3k_{2} r^{m/2}}{1-r}.
$$
Thus
\begin{equation} \label{E:nablaW6}
\lim_{m\to\infty} \sum_{i=0}^{m-1}\Ex \left[ \log \frac{V^T_{i+1}(e)U_{m-i}(\te)}{ V^T_{i}(e)U_{m-i}(\te)}\right] = \Ex\left[ V^{T}(e) U(\te)\right],
\end{equation}
completing the proof of Theorem \ref{T:computeSS3}.

\begin{Lem} \label{L:differentstart}
For any $u,u'\in\U$ and $e,e'\in \M$, if we let $e_{0},e_{1},\dots$ and $e'_{0},e'_{1},\dots$ be realisations of the Markov chain $P$ starting at $e_{0}=e$ and $e'_{0}=e'$ respectively. Then
\begin{equation} \label{E:diffstart}
\biggl|\Ex\left[\log\frac{\|X_{e_{m+1}}\cdots X_{e_{0}} u\|}{\|X_{e_{m}}\cdots X_{e_{0}} u\|} \right] - \Ex\left[\log\frac{\|X_{e'_{m+1}}\cdots X_{e'_{0}} u'\|}{\|X_{e'_{m}}\cdots X_{e'_{0}} u'\|} \right] \biggr|\le \frac{k_{2}D}{1-r} (m+1) (\xi \vee r)^{m},
\end{equation}
where $\xi$ and $D$ are the constants that satisfy \eqref{E:spectralgap}.
\end{Lem}

\begin{proof}
Using the maximal coupling, we create coupled versions of $(e,e')$, such that the coupling time $\tau$ satisfies
$$
\P\bigl\{ \tau\ge t \bigr\} \le \|\nu- P^{t}(e,\cdot)\|+\|\nu- P^{t}(e',\cdot)\|\le 2D \xi^{t}.
$$
Define
\begin{align*}
u_{\tau}:=\frac{X_{e_{\tau-1}}\cdots X_{e_{0}} u}{\|X_{e_{\tau-1}}\cdots X_{e_{0}} u\|},\qquad u'_{\tau}:=\frac{X_{e'_{\tau-1}}\cdots X_{e'_{0}} u'}{\|X_{e'_{\tau-1}}\cdots X_{e'_{0}} u'\|}.
\end{align*}
Then by the bound \eqref{E:rhobound},
\begin{align*}
\biggl|\Ex\left[\log\frac{\|X_{e_{m+1}}\cdots X_{e_{0}} u\|}{\|X_{e_{m}}\cdots X_{e_{0}} u\|} \right] -& \Ex\left[\log\frac{\|X_{e'_{m+1}}\cdots X_{e'_{0}} u'\|}{\|X_{e'_{m}}\cdots X_{e'_{0}} u'\|} \right] \biggr|\\
&\le\Ex\left[\biggl|\log\frac{\|X_{e_{m+1}}\cdots X_{e_{0}} u}{X_{e_{m}}\cdots X_{e_{0}} u\|} - \log\frac{\|X_{e'_{m+1}}\cdots X_{e'_{0}} u'\|}{\|X_{e'_{m}}\cdots X_{e'_{0}} u'\|} \biggr|\right] \\
&=\Ex\left[\biggl|\log\frac{\|X_{e_{m+1}}\cdots X_{e_{\tau}} u_{\tau}\|}{\|X_{e_{m}}\cdots X_{e_{0}} u\|} - \log\frac{\|X_{e_{m+1}}\cdots X_{e_{\tau}} u'\|}{\|X_{e_{m}}\cdots X_{e_{\tau}} u'_{\tau}\|} \biggr|\right]\\
&\le \Ex\left[ k_{2}r^{m-\tau}\right]\\
&\le \frac{k_{2}D}{1-r} \sum_{t=0}^{m} r^{m-t} \xi^{t}\\
&\le \frac{k_{2}D}{1-r} (m+1) (\xi \vee r)^{m}.
\end{align*}
\end{proof}

\end{document}